\documentclass[a4paper,preprintnumbers,floatfix,superscriptaddress,notitlepage,nofootinbib]{revtex4-2}

\usepackage{amsmath,amsthm,amsfonts,amssymb,times,bbm,graphicx,dsfont,stmaryrd}

\usepackage{physics}
\usepackage{hyperref}

\makeatletter

\renewcommand*{\p@subsection}{}

\renewcommand*{\p@subsubsection}{}
\makeatother

\newtheorem{theorem}{Theorem}

\newtheorem{lemma}[theorem]{Lemma}
\newtheorem{corollary}[theorem]{Corollary}
\newtheorem{definition}[theorem]{Definition}
\newtheorem{problem}{Problem}

\def\NN{\mathbbm{N}}

\def\Acal{\mathcal{A}}
\def\Bcal{\mathcal{B}}
\def\Ccal{\mathcal{C}}
\def\Dcal{\mathcal{D}}
\def\Ecal{\mathcal{E}}
\def\Fcal{\mathcal{F}}
\def\Gcal{\mathcal{G}}
\def\Rcal{\mathcal{R}}
\def\Ycal{\mathcal{Y}}
\def\Id{\mathds{I}}
\def\algtens{\otimes_{\text{alg}}}

\DeclareMathOperator{\Sym}{Sym}

\graphicspath{{Figures/}}

\begin{document}

\title{A convergent inflation hierarchy for quantum causal structures}

\author{Laurens T.\ Ligthart}
\email{ligthart@thp.uni-koeln.de}
\affiliation{Institute for Theoretical Physics, University of Cologne, Germany}
\author{Mariami Gachechiladze}
\affiliation{Department of Computer Science, Technical University of Darmstadt, Germany}
\affiliation{Institute for Theoretical Physics, University of Cologne, Germany}
\author{David Gross}
\affiliation{Institute for Theoretical Physics, University of Cologne, Germany}

\date{\today}

\begin{abstract} 
A \emph{causal structure} is a description of the functional dependencies between random variables. 
A distribution is \emph{compatible} with a given causal structure if it can be realized by a process respecting these dependencies. 
Deciding whether a distribution is compatible with a structure is a practically and fundamentally relevant, yet very difficult problem. 
Only recently has a general class of algorithms been proposed:
These so-called \emph{inflation techniques} associate to any causal structure a hierarchy of increasingly strict compatibility tests, where each test can be formulated as a computationally efficient convex optimization problem.
Remarkably, it has been shown that in the classical case, this hierarchy is \emph{complete} in the sense that each non-compatible distribution will be detected at some level of the hierarchy.
An inflation hierarchy has also been formulated for causal structures that allow for the observed classical random variables to arise from measurements on quantum states -- however, no proof of completeness of this \emph{quantum inflation hierarchy} has been supplied.
In this paper, we construct a first version of the quantum inflation hierarchy that is provably convergent.
It takes an additional parameter, $r$, which can be interpreted as an upper bound on the Schmidt rank of the observables involved.
For each $r$, it provides a family of increasingly strict and ultimately complete compatibility tests for correlations that are compatible with a given causal structure under this Schmidt rank constraint.
From a technical point of view, convergence proofs are built on \emph{de Finetti Theorems}, which show that certain \emph{symmetries} (which can be imposed in convex optimization problems) imply \emph{independence} of random variables (which is not directly a convex constraint). 
A main technical ingredient to our proof is a Quantum de Finetti Theorem that holds for general tensor products of $C^*$-algebras, generalizing previous work that was restricted to minimal tensor products.
\end{abstract}

\maketitle

\tableofcontents

\section{Introduction}\label{sec:introduction}

\subsection{Classical causal models}

Gaining information about causal relationships between variables from observational data is an important problem in empirical science  \cite{balke1997bounds,angrist1996identification,koller2009probabilistic}.
In the formalization laid out in Ref.~\cite{pearl2009causality}, causal relationships between classical random variables are modeled using  \emph{directed acyclic graphs} (DAGs; also: \emph{Bayesian networks} or \emph{causal structures}). 
Each vertex corresponds to a random variable.
Arrows denote causal relationships, in the sense that each variable is taken to be a function of its parents in the graph and independent randomness.
Causal structures that we consider here may have two types of vertices: variables that can be directly observed and variables that are not accessible, known as \emph{latent} or \emph{hidden} variables. 
In the graphical notation using DAGs, we will use circles to indicate latent variables and squares for observed ones (see Fig.~\ref{fig:triangle} for an example of a classical and a quantum causal structure).

We are interested in the following causal hypothesis testing problem:
Given a joint distribution over the observed random variables and a candidate causal structure,
can the distribution be realized in a model that is compatible with the structure?
One reason why this question is difficult is that we allow for \emph{arbitrary} functional relationships between parent and child variables, and also for the unobserved variables to take values in arbitrary sets.
(We do, however, restrict attention to the case where the observed variables take values in finite alphabets.)

Because there is an infinite set of possible functional relationships, it is a priori not obvious that the causal hypothesis testing problem is even algorithmically decidable.
Two computational solutions have recently been developed, though.
The first makes use of the realization that arguments based on Caratheodory's Theorem can be used to upper-bound the size of the sets in which the unobserved variables take values \cite{rosset2017universal}.
For finite-sized alphabets \emph{quantifier elimination} algorithms can in principle be used to decide compatibility (see, e.g.\ \cite{geiger2013quantifier,lee2017causal}).
The runtime of these solutions, however, renders these approaches impractical even for the smallest non-trivial scenarios.

The second solution is based on linear programming (LP) hierarchies \cite{wolfe2019inflation, navascues2020inflation}. 
Given a joint distribution $P$ and a \emph{relaxation level} $n$, 
these approaches construct an LP that runs in time polynomial in $n$ and the number of parameters describing $P$.
If the data fails the LP test at a level $k$, the distribution is not compatible with the causal structure.
If the data passes, no conclusions can be made.
As $n$ increases, more and more incompatible distributions will be recognized.
Such hierarchies are called \emph{complete} if every incompatible distribution will be rejected from some level $n$ onward.
Geometrically, the set of distributions accepted at relaxation level $n$ is an \emph{outer approximation} to the set of distributions compatible with a given structure.
A complete hierarchy thus corresponds to a sequence of outer approximations that converge to the true set.
In a break-through development, a convergent LP hierarchy for the classical causal hypothesis testing problem has been developed under the name of \textit{inflation technique}~\cite{navascues2020inflation}.

The high-level idea behind the inflation technique is to check for the existence of certain \emph{symmetric extensions}.
Indeed, assume that a distribution is compatible with a candidate causal structure.
One can then define an \emph{inflated} model that involves $n$ independent copies of the hidden variables.
This larger model has a number of symmetries:
One can exchange a hidden variable from one of the copies with the same hidden variable from another copy, without affecting the distribution (see Fig.~\ref{fig:triangle_inflation}).
The $n$-th level of the hierarchy tests whether such an $n$-fold inflated model exhibiting all these symmetries exists.

From a technical point of view, one difficulty of the causal compatibility problem lies in the fact that the distribution of the variables must not show any dependencies other than those that are explicitly modeled by the graph.
The set of independent probability distributions is not convex as a subset of all multivariate distributions.
Therefore, there is no direct way of enforcing the independence conditions that are part of the definition of a causal model in a convex optimization problem.
The inflation hierarchy circumvents this problem by imposing the symmetry constraints that follow from independence.
These symmetries are linear relations, which can easily be incorporated into a convex formulation.
Fortunately, it has long been realized that in an asymptotic sense, symmetries conversely often do imply independence.
Such results are known as \emph{de Finetti Theorems} \cite{de1974theory, raggio1989quantum, caves2002unknown, renner2007symmetry, brandao2017quantum, navascues2020inflation}. 
It is therefore unsurprising that the completeness proof of the classical inflation technique 
builds on a tailor-made de Finetti Theorem
\cite{navascues2020inflation},
and that a generalized Quantum de Finetti Theorem is at the heart of our own argument.

\begin{figure}
\centering
    (a)
    \includegraphics[width=.3\textwidth]{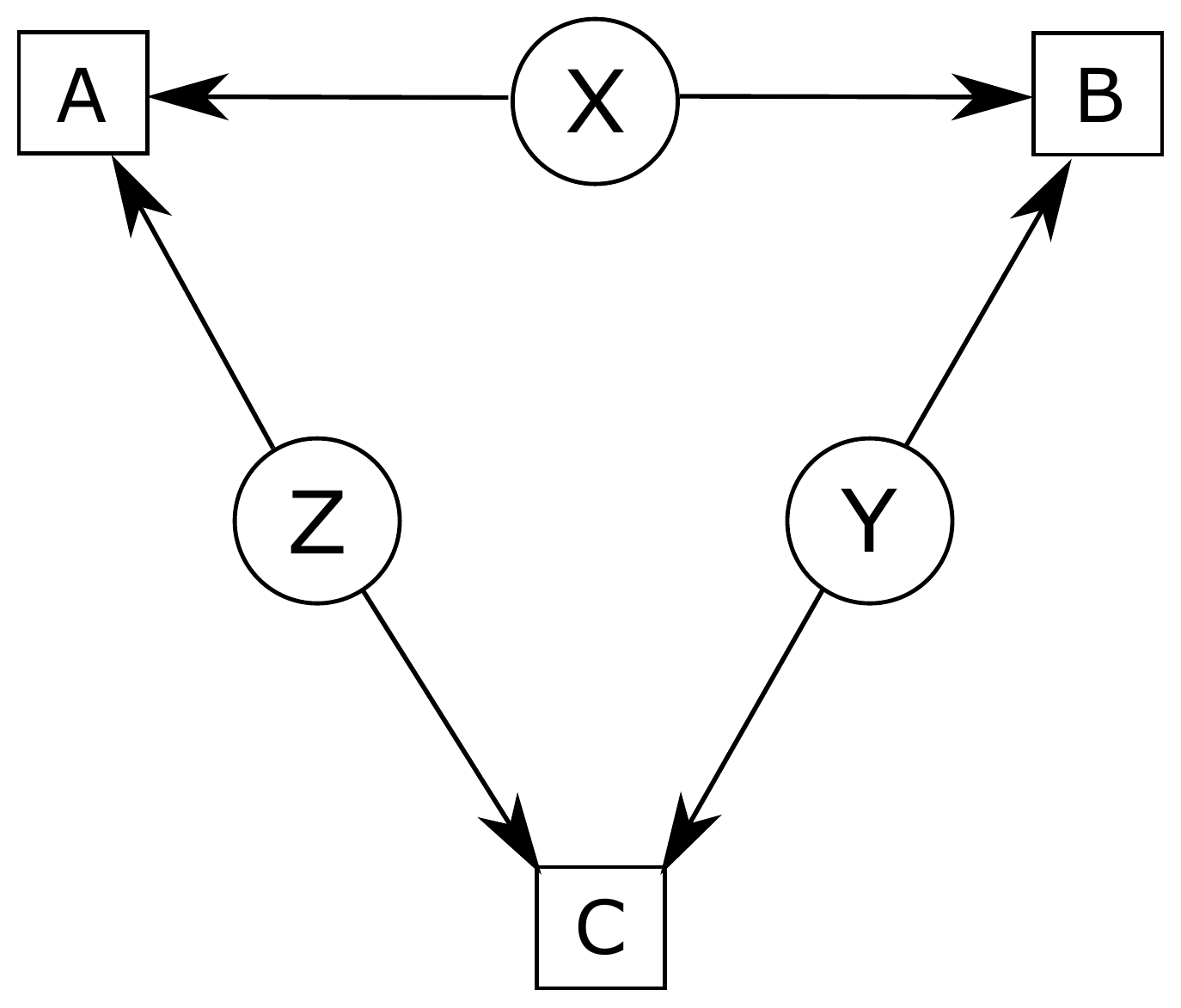}
    \hspace{2cm}
    (b)
    \includegraphics[width=.3\textwidth]{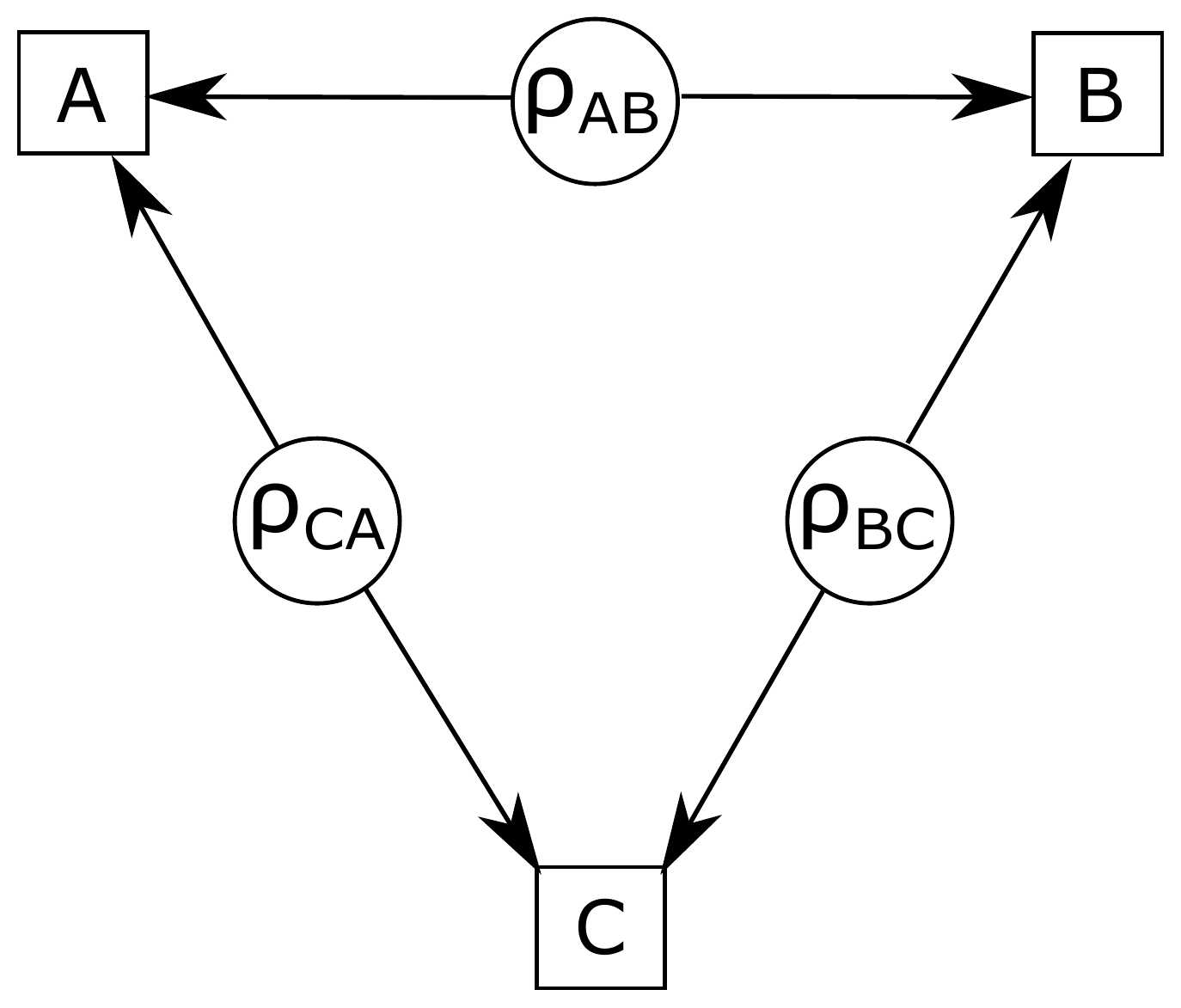}
    \caption{The \emph{triangle scenario} -- a conceptually simple, but mathematically non-trivial causal structure that serves as the guiding example in this paper.
	Round vertices denote \emph{latent variables} that are not directly accessible, while observed variables are written in squares.
	Arrows represent causal relations. 
	In the classical case depicted in Fig.~(a), the graph denotes the hypothesis that the observed variables $A, B, C$ arise in a process where 
	(1) the latent variables $X, Y, Z$ are chosen from a product distribution, and 
	(2) each observed variable $A, B, C$ is computed as a function of its graph-theoretical parents and additional independent randomness.
	Given a joint distribution of the observed variables and a candidate causal structure, we aim to decide whether the distribution is compatible with a process as outlined above.
	For example, take $A,B,C$ to be binary random variables.
	It is easy to see that a joint distribution where the outcomes are random but perfectly correlated is not compatible with the triangle scenario.
	In the quantum case, shown in Fig.~(b), each of the round nodes represents a bi-partite quantum state.
	One subsystem is distributed along each outgoing arrow.
	At each square vertex, a bi-partite measurement is performed on the two incoming quantum systems, and the result is assigned to an observed classical random variable.
}
        \label{fig:triangle} 
\end{figure}

\subsection{Quantum causal structures}
\label{sec:quantum causal}

It is natural to generalize the causal hypothesis testing problem to \emph{quantum causal structures} \cite{chaves2015information,wolfe2021quantum,costa2016quantum,barrett2019quantum,allen2017quantum}.
A conceptual difference to classical causal models is that, due to the no-broadcasting theorem, quantum states cannot be both measured and also serve as an input for further processing.
A general framework for a graphical notation for quantum causal structures is beyond the scope of this document (a detailed discussion can be found e.g.\ in Refs.~\cite{chaves2015information,allen2017quantum,barrett2019quantum}).
Here, we mainly focus on the subset of quantum causal structures known as \emph{correlation scenarios} \cite{fritz2012beyond}.
These comprise one layer of hidden nodes and one layer of observed nodes, with arrows pointing from hidden to observed ones (more general causal structures are discussed in Section \ref{sec:arbitrary_quantum_causal_structures}).

The input to the causal hypothesis test for correlation scenarios is a bipartite directed graph and a joint probability distribution with one classical variable corresponding to every observed node (see Fig.~\ref{fig:triangle}(b) and Fig.~\ref{fig:triangle_inflation}(b) for examples).
The problem is then to decide whether the classical distribution could have arisen from the following process:
\begin{enumerate}
\item
For each hidden node, prepare a quantum state on as many systems as there are outgoing arrows from that node. 
The quantum state can be entangled among the subsystems, but the states associated with different latent nodes must be independent.
Then, distribute the subsystems along the arrows to the observed nodes.
\item
At each observed node, perform a global measurement on all incoming quantum systems.
Assign the result to the observed random variable.
\end{enumerate}
This is the \emph{quantum causal hypothesis testing problem} we are concerned with in the present paper. 

To give an example, we again take the triangle scenario, which is depicted for the quantum case in Fig.~\ref{fig:triangle}(b). 
The latent variables are now quantum systems with quantum states that are labeled according to the observer to which they are sent. 
For example, $\rho_{AB}$ is a bipartite quantum state, of which the first part is sent to Alice and the second part to Bob. 
The arrows indicate independent quantum channels and whenever an arrow ends in a classical observable node, a measurement is performed at that node. 
Nodes that do not have any incoming edges are called root nodes and are assumed to be prepared in
independent initial states. 
Note the abuse of notation that is commonly used in the graphical representation of quantum causal structures: The hidden nodes are labeled by quantum states, as opposed to the quantum systems on which they live. This is in contrast to the classical case, where the hidden nodes are labeled by the random variables and not by their probability distributions. Here we have opted to adopt this commonly used abuse of notation, as the quantum state is generally considered as the more central object.

Recently, a hierarchy of semidefinite programming tests generalizing the inflation technique to the quantum case has been proposed \cite{wolfe2021quantum}. 
However, it is an open problem to decide whether it is \emph{complete}, i.e.\ whether every distribution that is incompatible with a given quantum causal structure will be rejected at some level of the hierarchy.

\subsubsection{Related problems}

In order to get a feeling for the causal hypothesis testing  problem, it is instructive to consider a few related examples. 

A closely related and particularly well-studied case is the \emph{Bell scenario} (Fig.~\ref{fig:Bell_scenario}). 
In a Bell scenario, observable classical variables come in pairs:
A parent variable that is interpreted as a party's choice of which measurement to perform,
and a child variable that is interpreted as the output of the measurement.
There is also one hidden variable that connects to every child variable.
The celebrated result of Bell says that the set of joint distributions observable in such a scenario differs between the case where the hidden variable is taken to be classical, and the case where it is allowed to be quantum. 
The causal hypothesis testing problem for Bell scenarios is well-understood, both classically and quantum mechanically.
In either case, the set of possible distributions of the measurement results conditioned on the input forms a convex set~\cite{brunner2014bell}.

If the hidden variable is classical, this convex set is a \emph{polytope}.
As such, it is defined by a finite number of linear inequalities known as \emph{Bell inequalities}.
This immediately gives rise to a finite algorithm for deciding the causal hypothesis test: A distribution is compatible with the causal structure if and only if it satisfies all Bell inequalities for the scenario.

Characterising the set of correlations that can arise from a shared quantum state is considerably harder.
In certain low-dimensional situations, complete analytical formulas are known~\cite{cirel1980quantum, masanes2005extremal}.
For the general case, the most explicit known numerical tool is the \emph{Navascués-Pironio-Acín} (NPA) hierarchy \cite{navascues2007bounding, navascues2015characterizing,pironio2010convergent},
a convergent\footnote{
	Convergence assumes that local subsystems are modeled as commuting observable algebras, see Sec.~\ref{sec:subsystems} for discussion.	
}
SDP hierarchy.

With the result presented in this paper, the situations for Bell scenarios and correlation scenarios are now analogous:
If the hidden variables are classical, then compatibility with both scenarios can be decided using linear programming;
if they are quantum, then there exist convergent SDP hierarchies.

\begin{figure}
    \centering
    \includegraphics[width=0.28\linewidth]{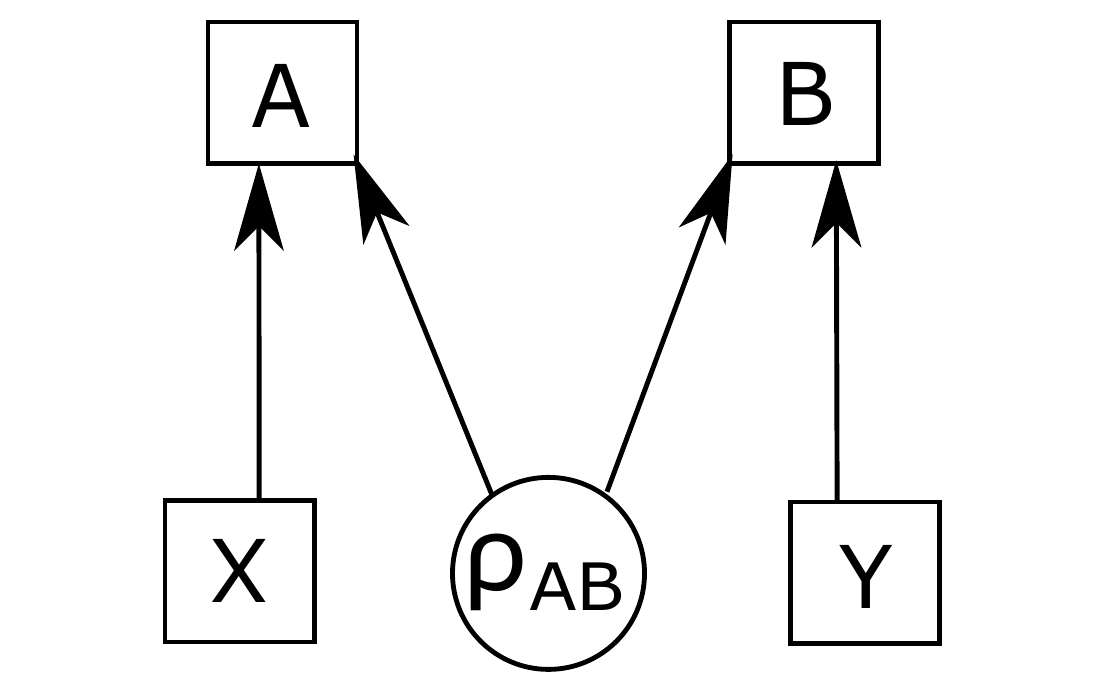}
    \caption{The Bell scenario, where the latent variable is a quantum state. The set of correlations that can arise from this quantum version of the Bell scenario is larger than its classical counterpart. This set of quantum correlations can be categorized by the NPA hierarchy, a converging hierarchy of semidefinite programs.}
    \label{fig:Bell_scenario}
\end{figure}

Finally, we remark that mathematical methods different from convex optimization hierarchies have been developed for the quantum causal compatibility problem.
They are usually computationally cheaper, but fail to be complete and might only be applicable to certain scenarios.
These include: 
polynomial Bell inequalities tailored for quantum networks~\cite{rosset2016nonlinear,chaves2016polynomial,tavakoli2021bell}; 
methods using entropy inequalities \cite{chaves2015information,henson2014theory,weilenmann2017analysing,budroni2016indistinguishability};
covariance matrices \cite{kela2019semidefinite,aaberg2020semidefinite,kraft2021characterizing};
methods tailored to  bilocality~\cite{branciard2012bilocal} and star network scenarios~\cite{tavakoli2014nonlocal}; 
correlations obtained in the triangle, when each source produces the same pure entangled state of two qubits~\cite{renou2019genuine}.

\subsection{Results}

Our main motivation is the approximate quantum causal compatibility problem.

\begin{problem}[Approximate quantum causal compatibility] \label{prob:compatibility}
	Given $\epsilon \geq 0$, a causal structure and a probability distribution over observable variables $P$, determine whether there exists a distribution $\tilde{P}$ that can be produced by a quantum model compatible with the causal structure, such that $\|\tilde{P}-P\|_2^2 \leq \epsilon$. 
\end{problem}

Following Refs.~\cite{navascues2020inflation, wolfe2021quantum}, we will phrase it as a special case of a more general \emph{causal polynomial optimization problem}.
There, the goal is to minimize a polynomial function $f_0(\rho)$ over states $\rho$ that are compatible with the causal structure and in addition satisfy a set of polynomial constraints $f_i(\rho)=0$.

The notion of a ``polynomial function in states'' needs to be explained.
A quantum model associates with each party an algebra of observables.
We assume that these algebras possess a finite set of \emph{generators} $\Gcal = \{g_i\}_i$ (in the main text, we describe how to derive a suitable set of generators for each causal structure and any given number of outcomes per party).
For now, we treat the $g_i$ as (non-commutative) variables -- later we will optimize over all possible assignments of concrete operators to the $g_i$.
Let $\Fcal^{(k)}(\Gcal)$ be the vector space formed by complex linear combinations of words of length $k$ in the symbols $g_i$ and $g_i^*$.
Each element $x\in\Fcal^{(k)}$ defines a ``linear function $f$ on states'' in the following sense:
Choose an assignment of operators to the $g_i$.
Let $\Dcal$ be the resulting algebra of observables.
Then $x$ can be understood as an element of $\Dcal$.
Let $\rho$ be a state on $\Dcal$.
Then the linear function is just $f: \rho \mapsto \rho(x)$.
To define a degree-$g$ polynomial, start with an element $x$ in the $g$-fold tensor product $\Fcal^{(k)} \otimes \dots \otimes \Fcal^{(k)}$ and set
\begin{align}\label{eqn:intro polarization}
	f: \rho \mapsto \rho^{\otimes g}(x).
\end{align}
We say that $x$ is a \emph{polarization} of the polynomial $f$.
Note that polarizations $x$ are defined independently of any assignments of concrete operators to the generators -- so they can be used to specify polynomial objective functions for problems that optimize over such assignments.\footnote{
		In the context of quantum inflations, there are two distinct sources of ``polynomials'' that must not be confused.
		First, general operators arise as (non-commutative) polynomials in the generators.
		The word ``polynomial'' in the \emph{non-commutative polynomial optimization} (NPO) framework refers to this sense.
		But NPO objective functions are still linear in the state.
		In contrast, the term ``polynomial'' in the \emph{quantum causal polynomial optimization problem} -- both as treated here and in Ref.~\cite{navascues2020inflation, wolfe2021quantum} -- indicates that we are allowing for objective functions that are polynomials in the state (in the sense explained above).
		Using the notions introduced in the main text, the degree of polynomial expressions in the generators of the algebra is connected to the level $k$ of the NPO hierarchy, while the degree of polynomial functions of the states corresponds to the level $n$ of the inflation hierarchy.
}

\begin{problem}[Quantum causal polynomial optimization] \label{prob:optimization}
	Given a causal structure, 
    the number of possible measurement outcomes for each party,
	a polynomial function $f_0$ on quantum states (as defined above), and a countable set of polynomial functions $f_1, f_2, \dots$ that are non-negative on quantum states compatible with the causal structure.
	Find
	\begin{align*}
	\begin{split} 
		f^\star &= \min_{\rho} f_0(\rho)  \\
			\text{s. t.} \quad &f_i(\rho) = 0 \quad i\geq 1 \\
			&\rho \text{ \rm{is compatible with the causal structure}}.
	\end{split}
	\end{align*}
\end{problem}

Problem \ref{prob:compatibility} reduces to Problem \ref{prob:optimization} by choosing $f_0$ to be the $2$-norm distance between the observed data $P$ and the one produced by the state.
More general objective functions can be important in applications -- see \cite[Section~VII]{wolfe2021quantum} for examples.

The rigorous definition of a quantum causal structure depends on a mathematical model of the notion of a ``local subsystem''.
In this paper, we model local subsystems via commuting observable algebras of bounded operators.
One could also consider alternative models:
Instead of via commuting algebras, locality could be formalized in terms of tensor products of Hilbert spaces 
\cite{scholz2008tsirelson,junge2011connes,fritz2012tsirelson,ji2020mip}.
Even though the causal compatibility problem only involves bounded operators (the algebra generated by the POVM (\emph{positive operator-valued measure}) elements that give rise to the observed probabilities), one could allow for unbounded operators in the local algebras on which they act.
A detailed discussion of these modeling decisions is given in Sec.~\ref{sec:subsystems}.
As we argue there, while these distinctions are certainly of mathematical interest, it seems unlikely that they will be relevant for data that arises from physical scenarios.

The SDP hierarchy that we describe differs slightly from the original quantum inflation hierarchy of \cite{wolfe2021quantum}.
Most importantly, we add two new parameters: $r, C$, which are related to the \emph{Schmidt decomposition} of the measurement operators.
To define them, consider a node of a quantum causal structure, say the one that gives rise to the random variable $A$ in the triangle scenario.
Each possible outcome is associated with a  
POVM element $E$.
As there are two incoming arrows to this vertex, $E$ acts on two quantum systems. 
Call the observable algebras acting on the respective subsystems $\mathcal{A}_-, \mathcal{A}_+$.
For fixed values of $r, C$, we assume that $E$ is of the form
\begin{align*}
    E &= \sum_{\alpha=1}^r 
		e_-(\alpha) e_+(\alpha) 
\end{align*}
for suitable operators 
$e_-(\alpha) \in \Acal_-,
e_+(\alpha) \in \Acal_+$
such that
\begin{align*}
	\|e_-(\alpha)\|, \|e_+(\alpha)\| \leq C.
\end{align*}
Call models like this \emph{rank-constrained with parameters $r,C$}.

\begin{problem}[Rank-constrained quantum causal polynomial optimization] \label{prob:optimization rank}
	With the notation of Problem~\ref{prob:optimization}, find
	\begin{align*}
	\begin{split} 
		f^\star_{r,C} &= \min_{\rho} f_0(\rho)  \\
			\text{s.\ t.} \quad &f_i(\rho) = 0 \quad i\geq 1 \\
			&\rho \text{ \rm{
   is a state on a model that is rank-constrained with parameters $r, C$ }} \\
   &\rho \text{ \rm{is compatible with the causal structure.
}}
	\end{split}
	\end{align*}
\end{problem}

Here, we construct semi-definite programming relaxations for these problems:

\begin{theorem}
	Use the notation of Problem~\ref{prob:optimization} and Problem~\ref{prob:optimization rank}.
	For every $r,C$, there exists a hierarchy of semi-definite programs indexed by an \emph{inflation parameter} $n$ and an \emph{NPO parameter} $k$.
	Denote the optimal values by $f^\star_{r,C,n,k}$.

	The hierarchy is complete in the sense that, as $n,k\to\infty$, the $f^\star_{r,C,n,k}$ converge to $f^\star_{r,C}$ from below.
	What is more, as $r,C\to\infty$, the $f^\star_{r,C}$ converge to $f^\star$ from above.
\end{theorem}

Increasing the parameter $C$ does not come with significant computational cost (c.f.\ the discussion of the \emph{Archimedean property} in \cite{pironio2010convergent}).
Larger values of $r$, in contrast, do correspond to a larger number of variables and constraints in the SDP formulation.
The decision to add these additional degrees of freedom must therefore be well-justified.
While we cannot prove that they are strictly necessary (which would in particular imply that the original quantum inflation hierarchy is not convergent), we identify some challenges that any constructive convergence proof that does not include these extra variables would face (Sec.~\ref{subsec:local_algebras}).

\subsubsection{Auxiliary results}

A quantum causal structure imposes independence conditions on the latent nodes. 
We deduce independence from the symmetry constraints imposed in the SDP hierarchy via a Quantum de Finetti Theorem.
A technical challenge arises because SDP hierarchies model local subsystems via commuting observable algebras, rather than via Hilbert space tensor products (also known as \emph{minimal tensor products} in the $C^*$-algebraic literature) \cite{scholz2008tsirelson,junge2011connes,fritz2012tsirelson,ji2020mip}.
To the best of our knowledge, the existing literature on de Finetti Theorems for infinite systems is phrased only in terms of the minimal tensor product \cite{stormer1967symmetric, raggio1989quantum}, which is not general enough for our purposes.
To address this, we show that the arguments in Ref.~\cite{raggio1989quantum} generalize to the \emph{maximal tensor product}, and hence
to any way of realizing local observables as commuting subalgebras of a global system.
This is the content of Theorem~\ref{thm:finetti maximale}.

Both the original quantum inflation hierarchy \cite{wolfe2021quantum} and our work are built 
on the \emph{non-commutative polynomial optimization} (NPO) hierarchy introduced in Ref.~\cite{pironio2010convergent}.
Because the generalized Quantum de Finetti Theorem that is central to our convergence proof requires $C^*$-algebraic methods, we rephrase the framework of Ref.~\cite{pironio2010convergent} in this language in Section~\ref{sec:cstar}.
Following Ref.~\cite[Section II.8.3]{blackadar2006operator}, we give a description of NPO problems as optimizations over the state space of a \emph{universal $C^*$-algebra}.
This more abstract formulation might be beneficial in arguments that require algebraic methods.

\subsection{Outline}

The paper is structured as follows. We  start in Sec.~\ref{sec:background_challenges} by explaining the necessary theory on $C^*$-algebras, causal structures and the inflation technique. Here we also outline the main challenges that one faces when trying to tackle the question of causal compatibility. In Sec.~\ref{sec:max_de_finetti} we solve one of the challenges by adapting a proof of a de Finetti Theorem for the minimal $C^*$ tensor product by Raggio and Werner \cite{raggio1989quantum} to a similar de Finetti Theorem for the maximal $C^*$ tensor product. 
Sec.~\ref{sec:new_SDP} describes an SDP hierarchy for the causal optimization problem with Schmidt-rank constraints on the measurement operators. We show that this hierarchy is convergent and has the causal compatibility problem as a special case. The proof relies heavily on the de Finetti Theorem of Sec.~\ref{sec:max_de_finetti}. While formulating the SDP hierarchy and the proof, we focus on the well-known triangle scenario, but in Sec.~\ref{sec:arbitrary_quantum_causal_structures} we show that a converging SDP hierarchy can be found for any quantum causal structure. We end the main text of the paper in Sec.~\ref{sec:Conclusions} by recounting the most important results of the paper and discussing some questions that remain open. 

\section{Technical background and challenges} \label{sec:background_challenges}

In this section, we introduce the technical tools used in this paper and explain the challenges one encounters when trying to rigorously formulate and decide the  completeness problem for inflation hierarchies.

\subsection{Challenge 1: Mathematical models of subsystems}
\label{sec:subsystems}

The definition of a quantum causal structure depends crucially on the notion of a \emph{subsystem}.
Here, we describe two subtle modeling decisions that have to be made when making this term precise.

\subsubsection{Hilbert space tensor products vs commuting observable algebras}

In elementary quantum mechanics, the central object that characterizes a quantum system is its Hilbert space.
In this framework, one thus associates to each subsystem a Hilbert space $\mathcal{H}_i$ and takes the joint Hilbert space to be their tensor product $\mathcal{H}_{12}=\mathcal{H}_1\otimes \mathcal{H}_2$.
The set of observables is then derived from the Hilbert space structure.
For the individual subsystems, the observables are the linear operators $\mathcal{A}_i = L(\mathcal{H}_i)$.
They can be embedded into $\mathcal{A}_{12}=L(\mathcal{H}_{12})$ by taking the tensor product with identities on the other subsystem:
\begin{align}\label{eqn:intro_min}
				\mathcal{A}_1 \simeq \mathcal{A}_1 \otimes \Id,
				\qquad
				\mathcal{A}_2 \simeq \Id \otimes \mathcal{A}_2.
\end{align}

In contrast, in \emph{algebraic quantum mechanics} (c.f.\ \cite[Chapter~8]{strocchi2008introduction}, \cite{bratteli2012operator}),
the set of observables is seen as being more central.
Consequently, one associates an observable algebra $\mathcal{A}_i$ with each subsystem.
A joint system is then any algebra $\mathcal{A}_{12}$ that contains $\mathcal{A}_1$ and $\mathcal{A}_2$ as commuting subalgebras and is generated by them.
Clearly, the construction in (\ref{eqn:intro_min}) provides an example of algebras standing in such a relation -- but it turns out that there are more general scenarios that cannot be realized using Hilbert space tensor products.

Commutativity has physical consequences, e.g.\ in terms of joint measurability -- so if we accept the quantum-mechanical description of observable phenomena in terms of operators, we are forced to conclude that measurements in space-like separated regions are described in terms of commuting operators.
However, the stronger requirement that the underlying Hilbert space forms a tensor product is not obviously physically motivated.

For a long time, it was an open problem to decide whether 
there are correlations that can be realized by performing measurements on commuting operators, but not on operators acting on distinct factors of a tensor product Hilbert space.
In quantum information theory, this question has been known as \emph{Tsirelson's problem} and was shown to be equivalent to other long-standing open problems in operator theory, most notably \emph{Connes' embedding problem} \cite{scholz2008tsirelson,junge2011connes,fritz2012tsirelson}.
In a recent break-through result, these questions have been decided:
The commuting-operator model \emph{does} capture more general correlations than the tensor-product model \cite{ji2020mip}.

The above raises the question which of the two mathematical models to adopt.
Here, we take a pragmatic approach.
It has long been realized (and in fact, has historically triggered Tsirelson to speculate) that commutativity is easily encoded as a constraint in SDPs that give outer approximations to the set of quantum correlations \cite{navascues2008convergent}. 
The same is \emph{not} true for the tensor product property.
Since either model is legitimate, but one is a better fit for the SDP hierarchies we want to make a statement about, we opt for the approach in which locality is modeled by commutativity.

Therefore, in our work, we will assume throughout that one can associate an algebra of observables with each party and that these algebras commute.

Terms -- like \emph{product states} or \emph{separable states} -- that are commonly defined in quantum information theory with reference to a Hilbert space tensor product structure will be used in this paper in an analogous way that only relies on commutativity.
In particular:

\begin{definition}
	Let $\Acal, \Bcal$ be commuting algebras.
A state $\rho$ is said to be a \emph{product state} across $\Acal\mid \Bcal$ if
\begin{align}
    \rho(x\cdot y) = \rho(x) \rho(y) \qquad \forall x \in \Acal, y \in \Bcal.
\end{align}
A state is said to be \emph{separable} across $\Acal \mid \Bcal$ if it can be written as a convex combination of such product states, i.e.~if
\begin{align}
    \rho(x \cdot y) = \int \dd \mu (\sigma) \sigma(x \cdot y) =  \int \dd \mu (\sigma) \sigma(x) \sigma(y),
\end{align}
for some probability measure $\mu$ over product states.
\end{definition}

\subsubsection{Bounded vs unbounded local operators}

In this paper, we are interested in observable probabilities that describe measurements on quantum systems.
Probabilities are associated with 
POVM elements.
POVM elements are bounded: Their operator norm does not exceed $1$.
It follows that the entire observable algebra generated by POVM elements consists of bounded operators.
Many problems -- e.g.\ the problem of characterizing the set of correlations compatible with a Bell scenario sketched above -- can be described solely in terms of this algebra.
From a technical point of view, this property can provide significant simplifications.
For example, 
the convergence proofs of the NPO hierarchy \cite{pironio2010convergent} or  
the Quantum de Finetti Theorem for infinite-dimensional quantum systems \cite{raggio1989quantum} make central use of the fact that operators are bounded.

It thus comes as bad news that this simplifying property is not obviously available for the causal compatibility problem.

Indeed, consider a node of a quantum causal structure, say the one that gives rise to the random variable $A$ in the triangle scenario.
Each possible outcome $A=a$ is associated with a POVM element $E_a$.
As there are two incoming arrows to this vertex, $E_a$ acts on two quantum systems. 
We therefore assume that the observable algebra $\Acal$ of the joint system is generated by two commuting subalgebras $\mathcal{A}_-, \mathcal{A}_+$.
These local algebras play an important role in the definition of the causal structure:
It is with respect to them that the state is required to factorize.
But, while $E_a$ is bounded, the authors are not aware of any result that would imply that one can assume the same is true for elements of $\mathcal{A}_-, \mathcal{A}_+$.

More concretely, we cannot exclude the possibility that there is a mathematical model of ``local quantum systems'' in which one can assign a precise meaning to the series
\begin{align}\label{eqn:bounded issue}
    E_a &= \sum_{\alpha=1}^\infty 
		e_-(a,\alpha) e_+(a,\alpha) 
\end{align}
for suitable \emph{unbounded} operators $e_-(a,\alpha) \in \Acal_-, e_+(a,\alpha) \in \Acal_+$, but where no such expression for $E_a$ exists if the $e_-, e_+$'s are required to be bounded.

In our precise definition of a quantum causal model, we will \emph{assume} that it is not necessary to allow for such singular situations.
The convergence proof makes use of this assumption (implicitly, by virtue of being phrased in terms of $C^*$-algebras, which model bounded operators).
In fact, the main difference between our hierarchy and the original quantum inflation one \cite{wolfe2021quantum} is that we add explicit generators and norm constraints for these local operators
(see Section~\ref{subsec:local_algebras} for further evidence that some such addition may be necessary).

While it is an interesting question about operator algebras whether the assumption is actually necessary, it seems that under mild \emph{physical} conditions, observed correlations can be approximated using models for which it is valid.
For example, if each subsystem is endowed with a non-degenerate Hamiltonian and the state has finite energy, one can always compress the local observables to finite-dimensional low-energy subspaces on which they are obviously bounded.
So as long as not \emph{both} the observable and the state display rather singular behavior, an approximate bounded model should always be possible in physical situations.

\subsection{\texorpdfstring{$C^*$}{}-algebras}
\label{sec:cstar}

\subsubsection{Definitions and the GNS construction}

The mathematical abstraction of an algebra of bounded operators on a Hilbert space is captured by the concept of a \emph{$C^*$-algebra}. 
In this section, we introduce the notions that are used in this paper. For more details we refer the reader to e.g.\ Refs.~\cite{blackadar2006operator, bratteli2012operator, takesaki1}.

Consider a complex algebra $\Acal$ with an anti-linear involution $*$.
A \emph{$C^*$-norm} on $\Acal$ is a norm satisfying
\begin{align} \label{eq:c*norm}
	  \norm{x y} \leq \norm{x}\norm{y},
		\quad
    \norm{x^* x} = \norm{x}^2 \qquad \forall\,x, y \in \Acal.
\end{align}
$\Acal$ is a \emph{$C^*$-algebra} if it is complete with respect to a $C^*$-norm.
A $C^*$-algebra is \emph{unital} if it contains the identity $\Id$.
In this paper, we only consider unital algebras and will no longer mention this attribute explicitly.
For example, the set of bounded operators on a Hilbert space together with the operator norm and the involution given by the adjoint realizes a $C^*$-algebra. 

Conversely, any abstract $C^*$-algebra can be realized as bounded operators on a Hilbert space.
To see how, we need to introduce the notion of a \emph{state}.
An element $x$ of a $C^*$-algebra $\mathcal{A}$ is \emph{positive} if it is of the form $x = y^* y$ for some $y\in\Acal$.
A state $\rho$ on a $C^*$-algebra $\Acal$ is a linear functional that is \emph{positive} in the sense that $\rho(x)\geq 0$ for all positive $x\in\Acal$ and which is \emph{normalized} in that  $\rho(\Id) = 1$. 
We denote the state space of $\Acal$ by $K(\Acal)$. 

Any state $\rho$ induces a sesquilinear form on the algebra $\Acal$ via $\langle x, y\rangle := \rho(x^\dagger y)$.
The form turns the  quotient $\Acal/\{ x \mid \langle x, x\rangle = 0 \}$ into a pre-Hilbert space on which $\Acal$ acts in a natural way.
The \emph{Gelfand-Naimark-Segal (GNS) construction} makes this observation precise and associates with every state $\rho$ a Hilbert space $\mathcal{H}_\rho$, a representation $\pi_\rho:\mathcal{A}\to\mathcal{H}_\rho$, and a vector $|\Omega\rangle$ such that
\begin{align*}
	\rho(x) = \expval{\pi_\rho(x)}{\Omega} \qquad\forall x\in\Acal
\end{align*}
Using the GNS construction, one can prove that any $C^*$-algebra can be realized as an operator algebra acting on a Hilbert space, providing a converse to the motivating example  of a $C^*$-algebra above.

\subsubsection{\texorpdfstring{$C^*$}{}-algebras from generators and relations}
\label{sec:generators}

To reason numerically about observable algebras, one needs to express them in a format that can be processed by a computer.
Both the original quantum inflation hierarchy \cite{wolfe2021quantum} and our work are built 
on the \emph{non-commutative polynomial optimization} (NPO) hierarchy introduced in Ref.~\cite{pironio2010convergent}.
There, algebras are specified by \emph{generators} and \emph{relations}.
We introduce the theory in this subsection.

Because the generalized Quantum de Finetti Theorem that is central to our convergence proof requires $C^*$-algebraic methods, we rephrase the framework of Ref.~\cite{pironio2010convergent} in this language.
To this end, we employ the terminology of \emph{universal $C^*$-algebras} as described in \cite[Section II.8.3]{blackadar2006operator}.

Let $\mathcal{G}=\{g_i\}_i$ be a countable set of symbols.
Denote by $\Fcal(\Gcal)$ the free complex $*$-algebra generated by the elements of $\mathcal{G}$.
Put differently, $\Fcal(\Gcal)$ is the set of finite complex linear combinations of words in the symbols $g_i$ and $g_i^*$, with multiplication defined by concatenation of words.
Choose a countable set $\mathcal{R} \subset\mathcal{F}(\Gcal)$. 
We aim to define the ``largest $C^*$-algebra with generators $\mathcal{G}$, subject to 
the constraint that each $q\in\mathcal{R}$ is positive''.
We will refer to the elements of $\mathcal{R}$ as \emph{relations}.

To make this notion precise, 
define a \emph{representation of $(\mathcal{G}|\mathcal{R})$} to be a homomorphism $\pi:\mathcal{F}(\Gcal)\to\mathcal{B}(\mathcal{H})$ from the free algebra into the set of bounded operators on some Hilbert space $\mathcal{H}$, such that $\pi(q)$ is a positive operator for every $q \in \mathcal{R}$.
On $\mathcal{F}(\Gcal)$, define
\begin{align}\label{eqn:universal norm}
	\norm{x} := 
	\sup \{ \|\pi(x)\| \mid  \pi \text{ is a representation of } (\mathcal{G}|\mathcal{R}) \}.
\end{align}
Now assume that the relations imply $\norm{x} < \infty$ for all $x\in \mathcal{F}(\Gcal)$ (more on how we ensure this in practice below). 
In this case, $\norm{\cdot}$ is a seminorm on $\mathcal{F}(\Gcal)$.
The \emph{universal $C^*$-algebra on $(\mathcal{G}|\mathcal{R})$}, abbreviated as $C^*(\mathcal{G}|\mathcal{R})$, is then the completion of $\mathcal{F}(\Gcal)$ with respect to this $C^*$-seminorm.\footnote{
Recall that, despite the everyday connotations of the word, the process of ``completing'' a metric space does not necessarily only add elements to it.
The completion can be defined as the set of Cauchy sequences, with two considered equivalent if the norm of their differences converges to zero. 
Every $x\in\mathcal{F}(\Gcal)$ gives rise to an element in $C^*(\mathcal{G}|\mathcal{R})$, represented by the sequence that is constant and equal to $x$.
Two elements $x,y\in\mathcal{F}(\Gcal)$ induce the same element in $C^*(\mathcal{G}|\mathcal{R})$ if and only if $\norm{x-y}=0$.
Put differently, the completion adds elements to the quotient space $\Fcal(\Gcal)/\{ x \mid  \|x\|=0\}$.
}
We will not differentiate notationally between an element $x\in\mathcal{F}(\Gcal)$ and its image in the completion $C^*(\Gcal|\Rcal)$.

In our applications, we will consider two types of relations:

\paragraph{Equality constraints.}
Let $x,y\in\mathcal{F}(\Gcal)$ and assume that $\mathcal{R}$ contains both $x-y$ and $y-x$.
It then follows easily that $x=y$ in the universal $C^*$-algebra $C^*(\mathcal{G}|\mathcal{R})$.
We will always assume that $g_1=:\Id$ is constrained to commute with all the others and obeys $\Id x = x \Id = x$ for all $x \in \Fcal(\Gcal)$, so that $C^*(\mathcal{G}|\mathcal{R})$ is unital.

\paragraph{Norm constraints.}
Let $x\in\mathcal{F}(\Gcal), C\in\mathbb{R}_+$ and assume that $\mathcal{R}$ contains $C^2\Id - x^*x$.
Then  $\norm{x} \leq C$ in  $C^*(\mathcal{G}|\mathcal{R})$.

We can now return to the assumption that $\|x\|<\infty$ for all $x\in\mathcal{F}$.
First, note that it suffices to ensure that the bound holds for every generator.
In quantum applications \cite{navascues2007bounding, navascues2015characterizing, wolfe2021quantum}, one is often interested in algebras where the generators represent projections, i.e.\ satisfy $g_i=g_i^*=g_i^2$.
Non-trivial projections in a $C^*$-algebra always have norm equal to 1.
The modified quantum inflation hierarchy introduced in this paper makes use of generators that need not be projections.
For those, we must include explicit norm constraints $\|g_i\|\leq C$.
(A similar approach is taken in Ref.~\cite{pironio2010convergent} to achieve what is referred to as the \emph{Archimedean property} there.)

We will need the following lemma, which justifies our characterization of 
$C^*(\Gcal|\Rcal)$ as the ``largest $C^*$-algebra such that each $q\in\Rcal$ is positive''.

\begin{lemma}\label{lem:all is positive}
	In $C^*(\Gcal|\Rcal)$, any element $q\in\Rcal$ is positive.
\end{lemma}

\begin{proof}
	We first show that for every representation $\phi$ of $C^*(\Gcal|\Rcal)$, it holds that $\phi(q)$ is a positive operator for each $q\in\Rcal$.
	(Using the terminology introduced above, this says that a representation of $C^*(\Gcal|\Rcal)$ 
	is also a representation of $(\Gcal|\Rcal)$).

	Fix a $q\in\Rcal$.
	There is no loss of generality in assuming $\|q\|\leq 2$.

	It holds that $q=q^*$, because
	\begin{align*}
		\|q-q^*\| = \sup_{\pi\text{ representation of } (\Gcal|\Rcal)} \|\pi(q-q^*)\|  = 0,
	\end{align*}
	as $\pi(q)$ is positive (and hence self-adjoint) by definition of representations of $(\Gcal|\Rcal)$.

	For every representation $\pi$ of $(\Gcal|\Rcal)$,
	it holds that $\|\pi(\Id-q)\|\leq 1$  \cite[Proposition II.3.1.2(iv)]{blackadar2006operator}. 
	From the definition of the seminorm, this implies that $\|\Id-q\|\leq 1$.

	Now assume for the sake of reaching a contradiction that for some representation $\phi$ of $C^*(\Gcal|\Rcal)$, the operator $\phi(q)$ is not positive.
	Using again \cite[Proposition II.3.1.2(iv)]{blackadar2006operator}, 
	$\|\phi(\Id-q)\|>1\geq\|\Id-q\|$, which is a contradiction, as representations are norm-contractions.

	Next, let $\rho\in K(C^*(\Gcal|\Rcal))$.
	By the above, using the GNS representation,
	\begin{align*}
		\rho(q)=\langle\Omega|\pi_\omega(q)|\Omega\rangle \geq 0.
	\end{align*}
	Hence $q$ is positive by \cite[Corollary II.6.3.5]{blackadar2006operator}.
\end{proof}

\subsubsection{The NPO hierarchy} \label{subsubsec:NPO}

Given generators $\Gcal$ and relations $\Rcal$, we are interested in certain linear optimization problems over states on the algebra they generate.

In Ref.~\cite{pironio2010convergent}, the NPO problem is
phrased as an optimization over representations $\pi$ of $\Fcal(\Gcal)$ and normalized vectors $|\phi\rangle$ in the representation space.
Concretely, choose an element $y_0 \in\Fcal(\Gcal)$ and a countable set $\{y_1, y_2, \ldots\} = \Ycal \subset \Fcal(\Gcal)$ and consider:
\begin{align}\label{eqn:npo originale}
\begin{split}
	f^\star_{\mathrm{NPO}}=
	\min_{\pi, \ket{\phi}}\quad& \expval{\pi(y_0)}{\phi} \\
	\text{s.\ t.}\quad & \pi(q) \succcurlyeq 0 \qquad q\in\Rcal, \\
	& \expval{\pi(y_i)}{\phi} \geq 0 \qquad y_i\in\Ycal. 
\end{split}
\end{align}
We prefer to think of this problem more abstractly, as an optimization over the state space of the universal algebra $C^*(\Gcal|\Rcal)$:
\begin{align}\label{eqn:npo}
\begin{split}
	f^\star_{\mathrm{uni}}=
	\min_{\rho\in K(C^*(\Gcal|\Rcal))}\quad& \rho(y_0) \\
	\text{s.\ t.}\quad & \rho(y_i) \geq 0 \qquad y\in\Ycal.
\end{split}
\end{align}
We may write $\min$ instead of $\inf$, because the Banach-Alaoglu Theorem implies that the state space is weak$^*$-compact and thus that the infimum over states evaluated on any fixed element of the algebra is attained.
Following \cite[Section 3.6]{pironio2010convergent}, one can in addition impose constraints of the form $\rho(\cdot z)=0$ for a countable set of $z\in\Fcal(\Gcal)$. 
We have omitted this type of constraint from the discussion, as it is not needed for our use cases.

\begin{lemma}\label{lem:all the same}
	The solutions of (\ref{eqn:npo originale}) and (\ref{eqn:npo}) coincide.
\end{lemma}

\begin{proof}[Proof (of Lemma~\ref{lem:all the same})]
	Let $\omega$ be an optimizer of (\ref{eqn:npo}).
	Let $\pi_\omega$ be the GNS representation and $\ket{\Omega} \in\mathcal{H}_\omega$ the vector that implements $\omega$.
	By Lemma~\ref{lem:all is positive}, $(\pi_\omega, \Omega)$ is feasible for (\ref{eqn:npo originale}) and achieves the optimal value of (\ref{eqn:npo}).

	Conversely, let $(\pi, \phi)$ be an optimizer of (\ref{eqn:npo originale}).
	For $x\in\Fcal(\Gcal)$, define $\rho(x):=\expval{\pi(x)}{\phi}$.
	If the seminorm vanishes on $x$, $\|x\|=0$, then, in particular, $\|\pi(x)\|=0$ and hence $\pi(x)=0$.
	Thus, $\rho$ is constant on cosets of the ideal of elements on which the seminorm vanishes and 
	therefore well-defined as a functional on $C^*(\Gcal|\Rcal)$.
	As such, it is feasible for (\ref{eqn:npo}) and achieves the optimal value of (\ref{eqn:npo originale}).
\end{proof}

We now briefly describe the semidefinite programming hierarchy introduced in Ref.~\cite{pironio2010convergent} and sketch the completeness proof.

Of course, the difficulty in solving (\ref{eqn:npo}) lies in the fact that $C^*(\Gcal|\Rcal)$ is, in general, infinite-dimensional.
The broad idea behind the NPO hierarchy is to partition $\Fcal(\Gcal)$ into an increasing family of finite-dimensional subspaces $\Fcal^{(k)} \subset \Fcal(\Gcal)$. 
At the $k$-th level of the hierarchy, one imposes the conditions that $\rho$ be a state and that the relations be fulfilled only to the extent to which they can be expressed using elements from $\mathcal{F}^{(2k)}$.

To carry out this program, let $\mathcal{F}^{(k)}$ be the space of all elements $x\in\Fcal(\Gcal)$ that can be expressed as a polynomial in the generators and their adjoints of degree at most $k$.
We fix some basis $\{b^{(k)}_1, \dots, b^{(k)}_{d_k}\}$ of each $\mathcal{F}^{(k)}$.

Recall that a linear functional $\rho$ on $C^*(\Gcal|\Rcal)$ is a state if and only if $\rho(\Id)=1$ and $\rho(x^*x)\geq 0$ for all $x\in C^*(\Gcal|\Rcal)$.
We impose a related condition by demanding that the matrix $\Gamma^{(k)}$ with elements
\begin{align}
	\Gamma_{ij}^{(k)} = \rho\big({b_i^{(k)}}^* b_j^{(k)}\big),
	\qquad i,j = 1, \dots, d_k
\end{align}
be positive semidefinite and that $\Gamma_{1,1}^{(k)}=\rho(\Id)=1$.

Next, consider a relation $q\in\Rcal$.
Let $l$ be the smallest integer such that $q\in\Fcal^{(2l)}$.
We relax the requirement that $q$ be positive to demanding that the matrix
\begin{align*}
	(\Lambda_q^{(k)})_{ij}
	=
\rho({b_i^{(k-l)}}^* q b_j^{(k-l)}),
	\qquad i,j = 1, \dots, d_{k-l}
\end{align*}
be positive semidefinite.

Let $k_0$ be such that $x\in\Fcal^{(2k_0)}$. 
For each $k\geq k_0$, one thus arrives at a relaxation of Eq.~(\ref{eqn:npo}) in terms of the semidefinite program
\begin{align}\label{eqn:npo finite}
\begin{split}
	f^k = \min_{\rho\in (\mathcal{F}^{2k})^*}\quad
	& \rho(y_0), \\
	\text{s. t.} \quad 
	&\rho(\Id) = 1,\\
	& \Gamma^{(k)}\succcurlyeq 0, \\
	& \Lambda_q^{(k)}\succcurlyeq 0 \qquad q \in \Rcal \cap \Fcal^{(2k)},  \\
	&\rho(y_i) \geq 0 \qquad y_i \in \mathcal{Y}\cap\Fcal^{(2k)}.
\end{split}
\end{align}

The completeness result of Ref.~\cite{pironio2010convergent} states that, in
the case where $|\Ycal|\leq \infty$ is finite, the optimal values $f^k$ of the
relaxations (\ref{eqn:npo finite}) converge to $f^\star_{\mathrm{NPO}} =
f^\star_{\mathrm{uni}}=:f^\star$ from below.

\begin{lemma}
	The completeness result $\lim_{k\to\infty} f^k = f^\star$ extends to the case of a countably infinite number of inequality constraints $\rho(y)\geq 0$.
\end{lemma}

\begin{proof}
	Choose some enumeration $y_1, y_2, \dots$ for the countable set $\Ycal$.
	Let $\Ycal_s=\{y_1, \dots, y_s\}$. 
	Assume that (\ref{eqn:npo finite}) with $\Ycal$ replaced by $\Ycal_s$ is feasible for every $k, s$, with optimal value $f^k_s$.
	Using the convergence proof of Ref.~\cite{pironio2010convergent} and Lemma~\ref{lem:all the same}, there exists a sequence of states 
	$\rho^\star_s\in K(C^*(\Gcal|\Rcal))$ that are feasible for (\ref{eqn:npo}) with inequality constraints $\Ycal_s$ and attain $f^\star_s:=\lim_{k\to\infty} f^k_s$.
	By the Banach-Alaoglu Theorem, there is a convergent subsequence.
	Let $\rho^\star$ be its limit point.
	Then, for each $y_i\in\Ycal$, $\rho^\star(y_i) \geq 0$, as this constraint is fulfilled by all but a finite number of the $\rho^\star_s$.
	Thus, $\rho^\star$ is feasible for (\ref{eqn:npo}) with all inequality constraints $\Ycal$ taken into account
	and attains $f^\star=\lim_{s\to\infty}f^\star_s$. 
\end{proof}

\vspace{0.2cm}
\noindent \textbf{Remark.} If a space $N\subset\Fcal(\Gcal)$ of elements with vanishing seminorm is known, then one can replace $\Fcal(\Gcal)$ by the quotient space $\Fcal(\Gcal)/N$ in the constructions above, while retaining convergence.
This can result in significantly smaller matrices that need to be treated.
In particular, every equality constraint $x=y$ gives rise to an element $x-y\in N$.
\vspace{0.2cm}

\subsubsection{Tensor products} \label{subsubsec:tensor_products}

We return to the problem of precisely modeling the notion of ``locality''.
If $\Acal$ and $\Bcal$ are the $C^*$-algebras of observables on one subsystem each, then the joint system should come with an observable algebra $\mathcal{C}$ that contains copies of $\Acal$ and $\Bcal$ as commuting subalgebras and is generated by them.
Unfortunately, these two requirements are not quite enough to uniquely determine $\Ccal$.
To understand the freedom we have in defining the set of global observables, start with the \emph{algebraic tensor product} $\Ccal_0 = \Acal \algtens \Bcal$.
This is the $*$-algebra of elements $x$ of the form
\begin{align}
    x = \sum_i a_i \otimes b_i \qquad a_i \in \Acal,\ b_i \in \Bcal
\end{align}
with multiplication and involution in $\Ccal_0$ defined factor-wise as
\begin{align}
    (a_1 \otimes b_1) (a_2 \otimes b_2) = a_1 a_2 \otimes b_1 b_2, \qquad \qquad (a_1 \otimes b_1)^* = a_1^* \otimes b_1^*. \label{eq:tensor_alg}
\end{align}

To promote the $*$-algebra $\Ccal_0$ to a $C^*$-algebra $\Ccal$, we have to endow it with a norm satisfying the $C^*$-norm property in Eq.~\eqref{eq:c*norm} and complete it with respect to this norm. 
The choice of this norm is not unique \cite{takesaki1, blackadar2006operator}.
There are two distinguished norms: minimal and maximal, named-so because they constrain the value of any $C^*$-norm on the algebraic tensor product by
\begin{align}\label{eqn:no surprise}
  \|x\|_{\min} \leq \|x\| 
	\leq 
	\|x\|_{\max}.
\end{align}
In the more general case of $n$ tensor factors, they are defined via their respective values on elements of $\Acal_1 \algtens \ldots \algtens \Acal_n$ as
\begin{align}
	\|x_1\otimes \dots \otimes x_n\|_{\min}
	&=
	\sup \{ \|\pi_1(x_1)\|\dots\|\pi_n(x_n)\| \mid  \pi_i \text{ a representation of } \Acal_i \},
	\label{eqn:min norm}
	\\
	\|x_1\otimes\dots\otimes x_n\|_{\max}
	&=
	\sup \{ \|\pi(x_1\otimes\dots\otimes x_n)\| \mid  \pi \text{ a representation of } \Acal_1 \algtens \ldots \algtens \Acal_n \},
	\label{eqn:max norm}
\end{align}
where the suprema are taken over representations of the respective $C^*$-algebras as operator algebras on a Hilbert space. We denote the $C^*$-algebra generated by the tensor product of $\Acal$ and $\Bcal$ and completed with respect to the norm $\norm{.}_\gamma$ by $\Acal \otimes_\gamma \Bcal$.

If the $\Acal_i$'s arise as bounded operators on Hilbert spaces $\Acal_i = \Bcal(\mathcal{H}_i)$, the approach from elementary quantum mechanics corresponds to their natural embedding into $\mathcal{B}(\mathcal{H}_1\otimes\dots\otimes\mathcal{H}_n)$. 
The operator norm in this picture corresponds to the minimal tensor product.
In this way, the elementary approach reappears as a special case of the algebraic construction.

Convergence with respect to the maximal norm implies convergence for \emph{any} operator representation of the global observable algebra, according to the definition in  Eq.~(\ref{eqn:max norm}).
Thus, in order to obtain the most general results, we will focus on the maximal tensor product e.g.\ in Section~\ref{sec:max_de_finetti}.

\subsection{Description of quantum causal structures in the commuting-operator model}
\label{sec:quantum causal models}

We can now restate the definition of quantum causal structures given informally in Sec.~\ref{sec:quantum causal} in mathematically precise terms.
Here, and later for the proof of the main result, we restrict attention to the triangle scenario (Fig.~\ref{fig:triangle}(b)).
Sec.~\ref{sec:arbitrary_quantum_causal_structures} discusses how to generalize the result to arbitrary correlation scenarios and causal structures.

Let $A, B$ and $C$ be random variables. We say that a probability distribution $P(A,B,C)$ is \emph{compatible} with the quantum triangle scenario, if it can be realized in the following mathematical model.

Assume that there is a $C^*$-algebra $\Dcal$ that is generated by 
commuting subalgebras $\Acal, \Bcal, \Ccal$ 
that are each associated with a vertex of the triangle.
Each of the algebras $\Acal, \Bcal, \Ccal$ is in turn generated by two commuting subalgebras:
$\Acal$ by $\Acal_-, \Acal_+$;
$\Bcal$ by $\Bcal_-, \Bcal_+$; and
$\Ccal$ by $\Ccal_-, \Ccal_+$.
They model the observables measurable on the subsystems that enter the respective node from either side in the diagram
(that is,  $\Acal, \Bcal, \Ccal$ and $\Dcal$ are  $C^*$-tensor products in the sense of Sec.~\ref{subsubsec:tensor_products} -- but we take no stance on which particular one).
Next, we assume that there is a state $\rho$ on $\mathcal{D}$ that factorizes according to the causal structure:
\begin{align}\label{eqn:factorization}
				\rho(A_- A_+ B_- B_+ C_- C_+) 
				= 
				\rho(C_+ A_- ) 
				\rho(A_+ B_-)
				\rho(B_+ C_-)
\end{align}
where $A_-\in\Acal_-, A_+\in\Acal_+, B_-\in\Bcal_-$ and so on.
Finally, we assume that there are POVMs 
\begin{align}
				\{E_a\}_a \subset \Acal, \quad
				\{F_b\}_b \subset \Bcal, \quad
				\{G_c\}_c \subset \Ccal
\end{align}
such that the joint distribution can be realized as
\begin{align}
				P(a,b,c)  = \rho(E_a F_b G_c), \label{eq:quantum_measurement}
\end{align}
where $P(a,b,c)$ is short-hand for $P(A=a,B=b,C=c)$, which is the probability corresponding to the values $a,b,$ and $c$.

\vspace{0.2cm}
\noindent \textbf{Remark.} 
The GNS construction applied to two commuting algebras acting on a product state gives rise to a tensor product Hilbert space.
Thus, by (\ref{eqn:factorization}), there is no loss of generality in assuming that $\Acal = \Acal_-\otimes_{\min} \Acal_+$ and likewise for $\Bcal$ and $\Ccal$.
However, the same is \emph{not} true for the tensor products between $\Acal$, $\Bcal$, and $\Ccal$.
(Certainly the same argument doesn't apply -- as $\rho$ does not factorize as a state between the nodes.
One can combine the results from Refs.~\cite{fritz2012beyond, ji2020mip} to see that there are correlations $P$ that cannot be modeled using minimal tensor products between nodes at all).
This observation does not obviate the need for a generalized Quantum de Finetti Theorem, as we will apply it to the state that is extracted from the SDP hierarchy, and there is no semidefinite constraint that can express that $\Acal$ is a minimal, rather than a general, tensor product of its constituents.

\subsubsection{Quantum inflation}
\label{sec:intro:inflation}

\begin{figure}
    \centering
    (a)
    \includegraphics[width=.3\textwidth]{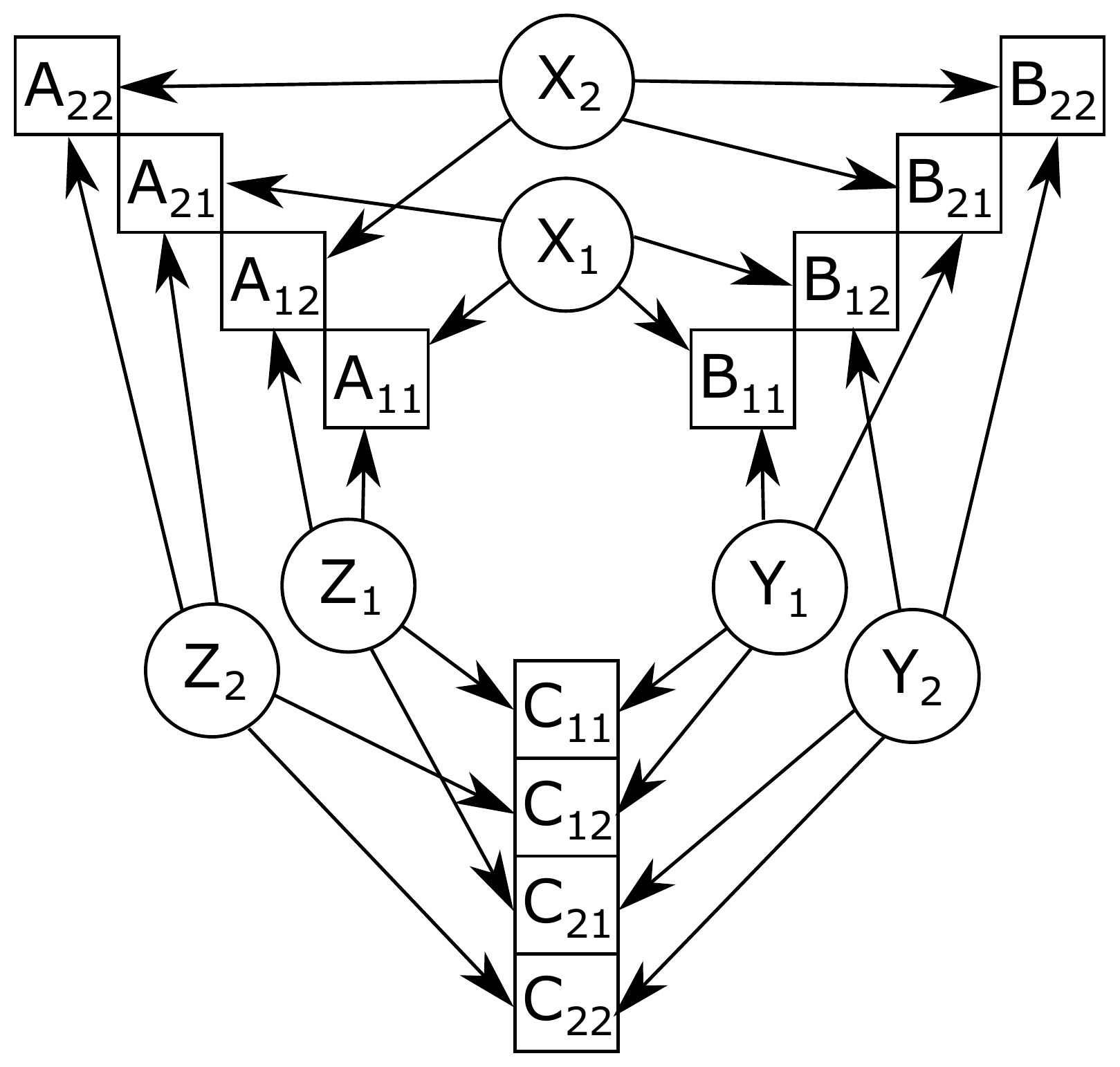}
    \hspace{2cm}
    (b)
    \includegraphics[width=.3\textwidth]{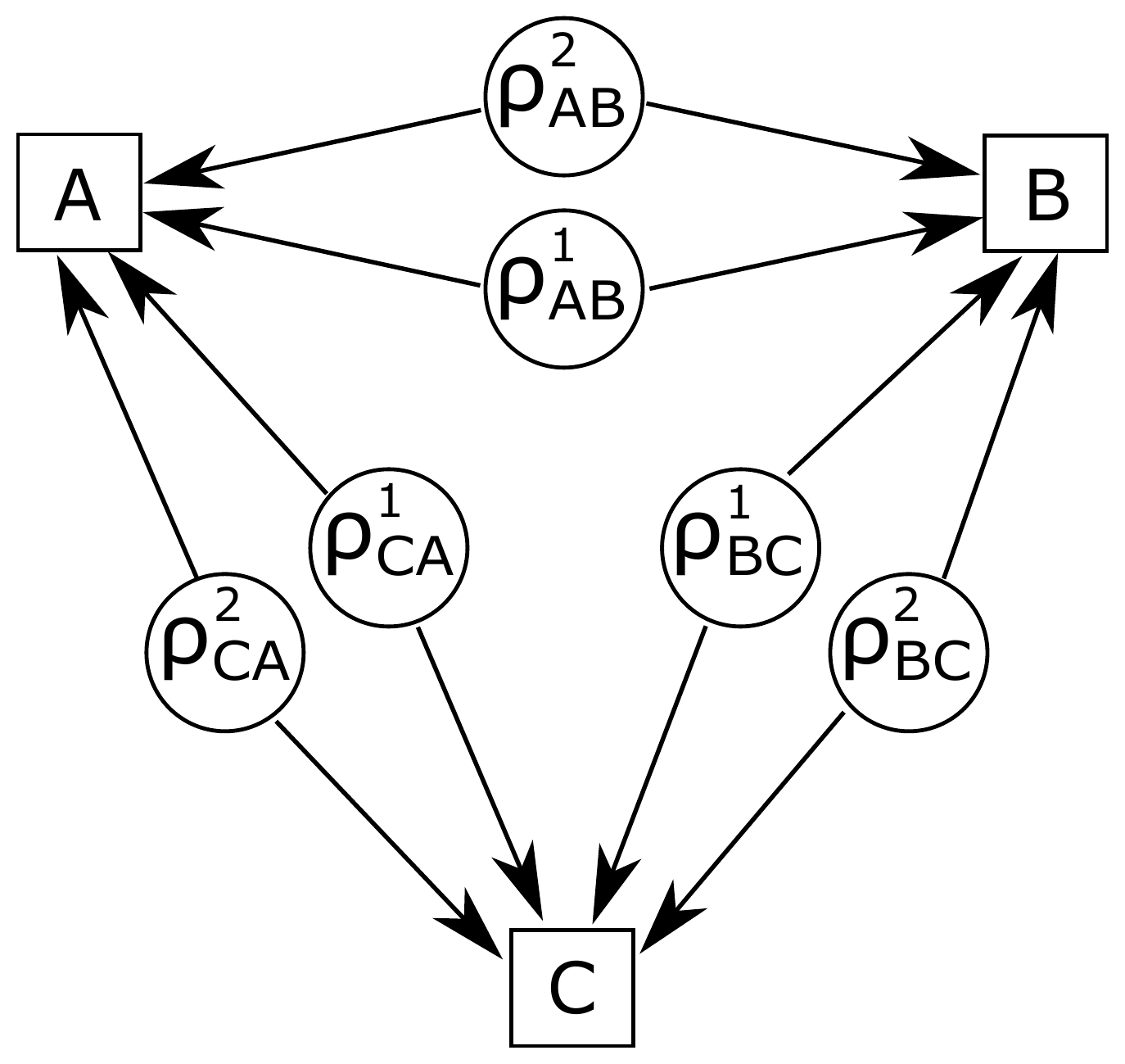}
    \caption{(a) The $n=2$ inflation of the classical triangle scenario. In the inflation procedure, the latent variables are copied and the observable variables are indexed according to these copies. For example, the variable $A_{12}$ is the result of $Z_1$ and $X_2$. (b) The $n=2$ inflation of the quantum triangle scenario. Again the latent variables are copied, but instead of obtaining copies of the observable variables, Alice, Bob and Charlie now have a choice of measurement operators. The measurements are performed over pairs of copies using the same measurement operators $\{E_a\}_a, \{F_b\}_b, \{G_c\}_c$, labeled by the copies they act on. For example, $E_a^{12}$ is acting on the parts of the quantum states $\rho_{CA}^1$ and $\rho_{AB}^2$ that are sent to Alice.}
    \label{fig:triangle_inflation}
\end{figure}

We give a brief overview of the quantum inflation technique, which was introduced in Ref.~\cite{wolfe2021quantum} and generalizes classical inflation 
\cite{wolfe2019inflation, navascues2020inflation}
to the case of quantum causal structures. 
This section serves to motivate the construction of the hierarchy in Sec.~\ref{sec:new_SDP} -- but in the rest of the paper, we will not rely on results and notation introduced here. 
We again focus on the triangle scenario (See Figs.~\ref{fig:triangle}(b) and \ref{fig:triangle_inflation}(b)).

Given a joint distribution $P(A,B,C)$, assume that there is a quantum model compatible with the triangle scenario.
Using the terminology of Sec.~\ref{sec:quantum causal models}, we thus know there exist a global observable algebra $\Dcal$ generated by local algebras
$\Acal_-, \Acal_+, \Bcal_-, \Bcal_+, \Ccal_-, \Ccal_+$, a state $\rho$ factorizing as in (\ref{eqn:factorization}), and POVM elements $\{E_a\}_a, \{F_b\}_b, \{G_c\}_c$ that reproduce the correlations $P$ as in (\ref{eq:quantum_measurement}).

Denoting the restrictions of 
$\rho$ 
to the algebra 
$\langle \Ccal_+ \cdot \Acal_-\rangle$ generated by $\Ccal_+, \Acal_-$ as $\rho_{CA}$,
to
$\langle \Acal_+ \cdot \Bcal_-\rangle$ as $\rho_{AB}$,
and
to
$\langle \Bcal_+ \cdot \Ccal_-\rangle$ as $\rho_{BC}$,
Eq.~(\ref{eqn:factorization}) is equivalent to demanding that $\rho$ factorizes as
\begin{align*}
	\rho = \rho_{CA}\otimes\rho_{AB}\otimes\rho_{BC}.
\end{align*}

For any \emph{level} $n$, we construct an \emph{inflated} model as follows.
Distribute $n$ independent copies of the original states $\rho_{AB}$, $\rho_{BC}$ and $\rho_{CA}$ among the three nodes $A$, $B$, and $C$. 
At each node, we consider $n^2$ POVMs
$\{E_a^{ij}\}_a, \{F_b^{kl}\}_b, \{G_c^{pq}\}_c$. 
The POVM element $E_a^{ij}$ replicates the original $E_a$, but acts on the $i$-th copy of the state $\rho_{CA}$ and $j$-th copy of the state $\rho_{AB}$. 
The other two cases are defined analogously. 
As a result, POVM elements $E_a^{ij}$, $F_b^{jl}$, and $G_c^{li}$ reproduce the original probabilities $P(a,b,c)$.

We now list a number of properties of the inflated model.
These properties can be directly imposed as constraints in an NPO program.
It follows that if $P$ is compatible with the causal model, then the resulting NPO problem is feasible for any inflation level $n$ 
\cite{wolfe2021quantum}.
In Sec.~\ref{sec:new_SDP} we construct a variant of this NPO hierarchy for which we supply a proof of the converse implication.

First, in Ref. \cite{wolfe2021quantum} it is assumed that the $\{E_a\}_a$, $\{F_b\}_b$ and $\{G_c\}_c$ are orthogonal projective measurements, rather than more general POVMs. This simplifies the SDP and can be done without loss of generality, because we do not restrict dimension and possible dilations would still be compatible with the causal structure.
\begin{align}
    &(E^{ij}_a)^* = E^{ij}_a &\forall i,j,a, \label{eq:proj_meas_1}\\
    &E^{ij}_a E^{ij}_{a'} = \delta_{a,a'} E^{ij}_a  &\forall i,j,a,a',\\
    &\sum_a E^{ij}_a = \Id &\forall i,j, \label{eq:proj_meas_3}
\end{align}
and similar for $\{F_b^{kl}\}$ and $\{G_c^{pq}\}$. In later sections we will drop this restriction, and will only assume that we have POVM elements, i.e.~non-negative operators that sum to the identity.

Second, operators that act on different subsystems commute:
\begin{align}
    &\left[ E^{ij}_a, F^{kl}_b \right] = \left[ E^{ij}_a, G^{kl}_c \right] = \left[ F^{ij}_b, G^{kl}_c \right] = 0 &\forall i,j,k,l,a,b,c, \label{eq:old_commute_1}\\
    &\left[ E^{ij}_a, E^{i'j'}_{a'} \right] = \left[ F^{ij}_b, F^{i'j'}_{b'} \right] = \left[ G^{ij}_c, G^{i'j'}_{c'} \right] = 0 & \forall i\neq i', j\neq j', a, a', b, b', c, c'. \label{eq:old_commute_2}
\end{align}

Third, there is a permutation symmetry, resulting from the fact that the global state is built out of independent copies of the original one.
For any polynomial $Q$ in the measurement operators $\{E_a^{ij}\}, \{F_b^{kl}\}, \{G_c^{pq}\}$ up to inflation level $n$ and for all permutations $\pi, \pi', \pi''$ of $n$ elements, the following must hold:
\begin{align} \label{eq:old_SDP_symmetry}
    \rho \Big(Q(\{E_a^{ij}, F_b^{kl}, G_c^{pq}\})\Big) = \rho \Big(Q(\{E_a^{\pi(i) \pi'(j)}, F_b^{\pi'(k) \pi''(l)}, G_c^{\pi''(p) \pi(q)}\})\Big).
\end{align}
For example, 
\begin{align}
    \rho(E^{11}_a F^{12}_{b} F^{21}_{b'} G^{21}_c ) =
    \rho(E^{12}_a F^{22}_{b} F^{11}_{b'} G^{21}_c ) = \rho(E^{12}_a F^{21}_{b} F^{12}_{b'} G^{11}_c ) =
    \rho(E^{22}_a F^{21}_{b} F^{12}_{b'} G^{12}_c ),
\end{align}
where we have swapped $\rho_{AB}^1 \leftrightarrow \rho_{AB}^2$ in the first step, $\rho_{BC}^1 \leftrightarrow \rho_{BC}^2$ in the second and $\rho_{CA}^1 \leftrightarrow \rho_{CA}^2$ in the third.

Fourth and finally, for the specific problem of causal compatibility the authors of Ref.~\cite{wolfe2021quantum} include constraints of the marginal distribution over $g\leq n$ copies of the triangle scenario. In particular, for the triangle scenario it must hold that
\begin{align} \label{eq:old_SDP_marginal}
    \rho \Big(\prod_{i=1}^g E_{a^i}^{ii} F_{b^i}^{ii} G_{c^i}^{ii} \Big) = \prod_{i=1}^g P(a^i, b^i ,c^i),
\end{align}
since these variables describe $g$ independent copies of the triangle causal structure. 
We will not be needing these types of constraints for our quantum inflation hierarchy, since we can already show convergence without them. Instead, for the approximate causal compatibility problem we will choose an objective function that, if the optimal value is $\epsilon$-close to 0, ensures that Eq.~\eqref{eq:old_SDP_marginal} approximately holds. If $\epsilon = 0$ Eq.~\eqref{eq:old_SDP_marginal} will hold exactly.

\subsection{Challenge 2: Infinite quantum de Finetti Theorem for general \texorpdfstring{$C^*$}{} tensor products}

In addition to the conceptual problems mentioned in Sec.~\ref{sec:subsystems}, switching to a more general notion of locality raises additional technical challenges.
Indeed, the basic idea underlying the quantum inflation hierarchy is to relax \emph{independence conditions} (which define the causal structure, but are non-convex and thus cannot directly be phrased as an SDP constraint) to \emph{symmetry conditions} (which are linear in elements of the algebra and easily incorporated into an SDP).
It is easily seen that independence implies symmetries in the inflated causal structure.
Central to convergence arguments is that sometimes, in an asymptotic sense, the converse is also true. This is the case for classical causal structures \cite{navascues2020inflation}.
Such converse results that obtain independence from symmetries are known as \emph{de Finetti Theorems} and have been formulated both for classical \cite{de1929funzione,diaconis1980finite} and for quantum \cite{stormer1967symmetric, raggio1989quantum,caves2002unknown,renner2007symmetry} probability theories.

Now a problem we face is that the literature on de Finetti Theorems for $C^*$-algebras to our best knowledge only pertains to minimal tensor products -- too narrow for our use case.
One of the main technical contributions of this work is the realization that the construction and proof of the Quantum de Finetti Theorem in \cite{raggio1989quantum} carries over from the minimal tensor product case for which it was formulated, to general $C^*$-tensor products.
Implementing this program is the role of Section~\ref{sec:max_de_finetti}.

\subsection{Challenge 3: Identifying the local observable algebras} \label{subsec:local_algebras}

We now present what we consider to be the most difficult challenge in deciding completeness of quantum inflation hierarchies.

Assume that a joined probability distribution $P(A,B,C)$ passes all levels of the original quantum inflation hierarchy, as outlined in Sec.~\ref{sec:intro:inflation}. 
We thus know that there is a $C^*$-algebra $\mathcal{D}$ generated by the observables $\{E_a^{ij}, F_b^{ij}, G_c^{ij}\}$ and a state $\rho$ that reproduces the observed correlations (Eq.~(\ref{eq:old_SDP_marginal})) and is symmetric (Eq.~(\ref{eq:old_SDP_symmetry})).

We now need to verify that this quantum model fulfills the causal constraints laid out in Sec.~\ref{sec:quantum causal models}.
This involves, in particular, showing that
\begin{enumerate}
	\item
one can embed the algebra $\langle \{ E_a^{ii} \}_a \rangle$ containing the measured POVM elements into a potentially larger algebra $\mathcal{A}^{ii}$ of all observables associated with the vertex $A$, such that
		$\mathcal{A}^{ii}$ is generated by two commuting subalgebras $\mathcal{A}_-^{i}, \mathcal{A}_+^{i}$, and
	\item
 that the state $\rho$ can be extended to all of $\mathcal{A}^{ii}$ and that it factorizes in the sense that for each $A_-\in\mathcal{A}_-^{i}, A_+\in\mathcal{A}_+^{i}$ we have $\rho(A_- A_+) = \rho(A_-)\rho(A_+)$.
\end{enumerate}
The second condition can be addressed using the generalized Quantum de Finetti Theorem presented below.
The first condition, however, seems much more challenging:
\emph{There is no obvious ansatz for constructing $\mathcal{A}^{ii}$ and its commuting generators $\mathcal{A}^{i}_-, \mathcal{A}^{i}_+$ from the algebra $\mathcal{D}$ that results from the original quantum inflation hierarchy of Ref.~\cite{wolfe2021quantum}}.
In fact, in the subsection just below, we will give an argument that suggests that $\mathcal{D}$ does not in general contain local observable algebras $\mathcal{A}^i_-, \mathcal{A}^i_+$ that satisfy the two conditions above.
It would then follow that if the original hierarchy is complete, any constructive proof of that fact would necessarily have to introduce additional operators that are not generated by the measured POVMs and their copies.

Our modified quantum inflation hierarchy follows such an approach (for details, see Sec.~\ref{sec:new_SDP}).
The modified hierarchy contains generators $e_-^i(a,\alpha), e_+^j(a,\alpha)$, which are constrained to commute unless both the upper and lower indices coincide.
One can then \emph{define}
\begin{align}
	\mathcal{A}_-^i &= \langle \{ e_-^i(a,\alpha) \}_{a, \alpha} \rangle, \\
	\mathcal{A}_+^i &= \langle \{ e_+^i(a,\alpha) \}_{a, \alpha} \rangle, \\
	\mathcal{A}^{ij} &= \langle \mathcal{A}_-^i \cdot \mathcal{A}_+^j \rangle,\\
    E_a &= \sum_{\alpha=1}^r e_-(a,\alpha) e_+(a,\alpha).
\end{align}
The observables at the other two vertices are treated analogously.
This modified hierarchy thus fulfills condition 1.~listed above by construction.
Theorem~\ref{thm:hierarchy_convergence} then shows that the generalized Quantum de Finetti Theorem implies that there exists a state $\rho$ such that condition 2.~holds as well.

\subsubsection{Example of measurement operators that do not generate elements from the local algebras} 

In this subsection, we provide evidence for the claim that the algebra $\mathcal{D}$ that results from the original quantum inflation hierarchy does not in general contain the local observable algebras satisfying the two conditions laid out above.
The purpose of the material presented here is to motivate our ansatz and guide future research -- it is not required to understand the rest of the paper.

The strategy is to give a natural example of mutually commuting observable algebras 
$\mathcal{A}^i_-, \mathcal{A}^j_+$ and POVMs $E^{ij}_a \in \mathcal{A}^i_-\cdot \mathcal{A}^j_+$
such that the algebra generated by the $\{E^{ij}_a\}_{ija}$ does not contain any non-trivial local observable, i.e.\ no element in any of the $\mathcal{A}^i_-$ or $\mathcal{A}^j_+$ other than $\Id$.
This does not constitute a \emph{proof} of the claim made above:
We do not know whether there are correlations $P(A,B,C)$ that will cause the original inflation hierarchy to output such a model.
But it does show that there are natural choices for the operators $E^{ij}_a$ that fulfill all the constraints of the hierarchy, while failing to generate the local observables with respect to which the factorization properties of the causal structure are defined.

The model is very simple:
For $i,j=1, \dots, n$, let $\mathcal{A}^i_-, \mathcal{A}^j_+$ be the observable algebra of one qubit each.
Consider the maximally entangled \emph{magic basis}
\begin{align*}
	\ket{\psi_1}^{ij} &= \frac1{\sqrt 2} \big( \ket{00} + \ket{11} )^{ij}, \\
	\ket{\psi_2}^{ij} &= \frac1{\sqrt 2} \big( \ket{01} + \ket{10} )^{ij}, \\
	\ket{\psi_3}^{ij} &= \frac1{\sqrt 2} \big( \ket{00} - \ket{11} )^{ij}, \\
	\ket{\psi_4}^{ij} &= \frac1{\sqrt 2} \big( \ket{01} - \ket{10} )^{ij},
\end{align*}
and define POVMs
\begin{align*}
	E^{ij}_a =  \dyad{\psi_a}^{ij}.
\end{align*}

\begin{lemma}
	The algebra $\mathcal{D}$ generated by $\{E^{ij}_a\}$ for $i,j=1, \dots n; a = 1, \dots, 4$ does not contain any non-trivial local operator.
\end{lemma}

\begin{proof}
	The magic basis is a stabilizer basis, and we can thus express the projection operator onto each vector by summing over the respective stabilizer group.
	In terms of the usual Pauli operators, this gives
	\begin{align} 
		E^{ij}_1 = \frac{1}{4}(\Id + X_-^{i} X_+^{j} + Z_-^{i} Z_+^{j} - Y_-^{i} Y_+^{j}),\label{eq:pauli_1} \\ 
		E^{ij}_2 = \frac{1}{4}(\Id + X_-^{i} X_+^{j} - Z_-^{i} Z_+^{j} + Y_-^{i} Y_+^{j}), \\
		E^{ij}_3 = \frac{1}{4}(\Id - X_-^{i} X_+^{j} + Z_-^{i} Z_+^{j} + Y_-^{i} Y_+^{j}), \\
		E^{ij}_4 = \frac{1}{4}(\Id - X_-^{i} X_+^{j} - Z_-^{i} Z_+^{j} - Y_-^{i} Y_+^{j}). 
		\label{eq:pauli_4}
	\end{align}
	Let
	\begin{align*}
		\bar X = \prod_{i}^n X^{i}_- X^{i}_+,
		\qquad
		\bar Z = \prod_{i}^n Z^{i}_-  Z^{i}_+.
	\end{align*}
	Because distinct Pauli operators on the same system anti-commute,
	\begin{align*}
		[\bar X, E^{ij}_a ]=[\bar Z, E^{ij}_a ]=0
	\end{align*}
	and thus $\mathcal{D}$ is contained in the commutant of $\bar X, \bar Z$.
	But there is no non-trivial local operator that commutes with both $\bar X$ and $\bar Z$.
\end{proof}

We note that Eqs.~\eqref{eq:pauli_1}-\eqref{eq:pauli_4} imply that the effects $E^{ij}_a$ have Schmidt-rank $\leq 4$ and a product decomposition with factors of operator norm $C\leq \frac14$.

\section{de Finetti Theorem for the maximal tensor product} 
\label{sec:max_de_finetti}

To the best of our knowledge, the existing literature on de Finetti Theorems for infinite systems is phrased only in terms of the minimal tensor product \cite{stormer1967symmetric, raggio1989quantum}.
These results are not directly applicable to the quantum models that result from the NPO hierarchy.
Indeed, the latter naturally guarantees the existence of a representation $\pi_\rho$ of the algebraic tensor product as operators on a Hilbert space that arise from a state $\rho$ on $\Acal_1 \algtens \ldots \algtens \Acal_n$ via the GNS construction.
While the resulting operator norm $\|\pi_\rho(x)\|$ constitutes a $C^*$-norm on the tensor product, we have no a priori control over its value beyond the constraints in Eq.~(\ref{eqn:no surprise}).

The purpose of this section is therefore to retrace the arguments given by Raggio and Werner in Ref.~\cite{raggio1989quantum} to verify that the infinite de Finetti Theorem established there generalizes to arbitrary $C^*$-norms on algebraic tensor products.
We also present a somewhat simpler formulation that is sufficient for our purposes.

In fact, we state the results only in terms of the  maximal $C^*$ tensor product norm.
A priori, it is possible that one can derive stronger results for the GNS norm $\|\pi_\rho(\cdot)\|$, in particular if the state $\rho$ is known to have symmetries.
We leave this possible improvement open for later investigations.

Let $\Dcal$ be a unital $C^*$-algebra and let
\begin{align*} 
	\Dcal^n = \Dcal^{\otimes_{\max} n}
\end{align*}
be the completion of the algebraic tensor product of $n$ copies of $\Dcal$ with respect to the maximal $C^*$-norm. 
The infinite maximal tensor product is defined as the \emph{inductive limit}
\begin{align}
    \Dcal^\infty = \lim_{n\rightarrow \infty} \Dcal^n.
\end{align} 
We recall the definition \cite[Section~II.8.2]{blackadar2006operator}:
For any $n,k\in\NN$, 
there is a natural embedding
\begin{align*}
	\Dcal^n \to \Dcal^{n+k}, \qquad x \mapsto x\otimes\Id^{\otimes k}.
\end{align*}
It allows us to define addition and multiplication between elements of the union
\begin{align}\label{eqn:local}
	\bigcup_{n=1}^\infty \Dcal^n
\end{align}
by embedding the element living in the smaller tensor power into the larger one and performing the operations there.
The resulting $*$-algebra is the \emph{local algebra}, called so as each of its elements lives in a finite tensor power.
The inductive limit $\Dcal^\infty$, the \emph{quasi-local algebra}, is the completion of the local algebra with respect to the $C^*$-norm
$\|\cdot\|_{\Dcal^\infty}$ on (\ref{eqn:local}) which assigns to every $x\in\Dcal^n$ the value
\begin{align*}
	\|x\|_{\Dcal^\infty} 
	= \|x\|_{\Dcal^n}
	= \lim_{k\to\infty} \|x\otimes\Id^{k}\|_{\Dcal^{n+k}}.
\end{align*}

We will not notationally distinguish between an element $x\in\Dcal^n$ and its embedding in $\Dcal^\infty$.
Note that any $x\in\Dcal^n$ and its extensions $x\otimes\Id^k$ are identified in $\Dcal^\infty$.

For any $n$ and permutation $\pi\in S_n$, there is an automorphism $\alpha_\pi$ on $\Dcal^n$ which acts by permuting tensor factors
\begin{align*}
	\alpha_\pi(x_1\otimes \dots \otimes x_n)
	=
	x_{\pi_1}\otimes\dots\otimes x_{\pi_n}.
\end{align*}
It extends to any $\Dcal^{n+k}$ by letting $\pi$ act on the first $n$ tensor factors, and by continuity to $\Dcal^\infty$.

A state $\rho$ on $\Dcal^\infty$ is \emph{symmetric} if $\rho(x) = \rho(\alpha_\pi(x))$ for every $x\in\Dcal^\infty$.
Denote the set of symmetric states by $K_s(\Dcal^\infty)$.
We aim to show:

\begin{theorem}[Max tensor product Quantum de Finetti Theorem]\label{thm:finetti maximale}
	Let $\rho\in K_s(\Dcal^\infty)$ be a symmetric state on an infinite maximal tensor product
	\begin{align*}
		\Dcal^\infty = \lim_{n\to\infty} \Dcal^{\otimes_{\max} n}.
	\end{align*}
	Then there exists a unique probability measure $\mu$ over states on $\Dcal$ such that 
	for all $x\in\Dcal^\infty$,
				\begin{align*}
					\rho(x) = \int_{K(\Dcal)} \Pi_\sigma(x) \,\mathrm{d}\mu(\sigma),
				\end{align*}
			    where $\Pi_\sigma$ is the infinite symmetric product state on $\Dcal^\infty$ associated with the state $\sigma$ on $\Dcal$.
\end{theorem}

Key to the proof is to show that symmetric states $\rho$ define a state on the \emph{abelian} $C^*$-algebra of \emph{symmetric observables} whose multiplication law is derived from the tensor product.
It is established in the general theory of $C^*$-algebras \cite[Chapter~4]{bratteli2012operator}
that pure states $\phi$ of abelian algebras are homomorphisms, i.e.\ that $\phi(xy) = \phi(x)\phi(y)$, and that general states of abelian algebras are unique convex combinations of pure states.
Using the fact that in our case, the product $xy$ is related to the tensor product $x\otimes y$, we will obtain the claimed decomposition of $\rho$ as a convex combination of symmetric product states.

To construct the symmetric algebra, define the
\emph{symmetrization map} 
\begin{align}
\begin{split}
    \Sym^n(x) 
		= \frac{1}{n!} \sum_{\pi \in S_n} \alpha_\pi (x),
\end{split}
\end{align}
and let $\Sym^n(\Dcal)$ be the image of $\Dcal^n$ under $\Sym^n$. 
Define the \emph{symmetric local algebra} to be the set
\begin{align}\label{eqn:symmetric local algebra}
	\bigcup_{n=1}^\infty \Sym^n(\Dcal)
\end{align}
with an associative and abelian multiplication law given by the \emph{symmetrized tensor product}
\begin{align*}
				\star: \Sym^n(\Dcal) 
				\times
				\Sym^m(\Dcal) 
				\to
				\Sym^{n+m}(\Dcal),
				\qquad
				x \star y = \Sym^{n+m}(x\otimes y).
\end{align*}

Our aim is to mimic the construction of $\Dcal^\infty$ to arrive at a \emph{symmetric quasi-local algebra} $\Sym^\infty(\Dcal)$.
To this end, define 
embeddings
\begin{align}\label{eqn:symmetric imbedding}
	\Sym^n(\Dcal) \to \Sym^{n+k}(\Dcal),
				\qquad
				x \mapsto x \star \Id^{\otimes k}.
\end{align}
As was the case for $\Dcal^\infty$, addition between two symmetric local elements can now be defined by embedding the lower power into the higher power and performing the addition there.
This convention turns the symmetric local algebra into an abelian $*$-algebra.
One can endow it with a $C^*$-seminorm \cite[Section~6.1]{murphy1990c} so that the completion $\Sym^\infty(\Dcal)$ is an abelian $C^*$-algebra:

\begin{lemma} \label{lem:Cstarnorm}
	The limit
	\begin{align}\label{eqn:symmetric norm}
		\|x\|_{\Sym} := \lim_{k\to\infty} \|x\star\Id^{\otimes k}\|_{\Dcal^\infty}
	\end{align}
	defines a $C^*$-seminorm on the symmetric local algebra $\cup_n \Sym^n(\Dcal)$
	fulfilling 
	\begin{align}\label{eqn:norm contraction}
		\|x\|_{\Sym} \leq \|x\|_{\Dcal^\infty}.
	\end{align}
	The completion $\Sym^\infty(\Dcal)$ is an abelian $C^*$-algebra.
\end{lemma}

The central ingredient to the proof is the following combinatorial lemma, which shows that the multiplication on $\cup_n \Sym^n(\Dcal)$ inherited from $\cup_n \Dcal^n$ and the newly defined $\star$-multiplication are asymptotically equivalent. 

\begin{lemma} \label{lem:combinatorics}
Let $x \in \Sym^m(\Dcal)$ and $y \in \Sym^n(\Dcal)$. 
Then
\begin{align*}
    \lim_{k \to \infty} \norm{(x\star \Id^{\otimes (k-m)})(y\star \Id^{\otimes (k-n)}) - (x \star y \star \Id^{\otimes (k-m-n)})  }_{\Dcal^\infty} &= 0.
\end{align*}
\end{lemma}

\begin{proof}
	Choose two sets of respective size $m, n$ uniformly at random from $[k]:=\{1, \dots, k\}$.
	The probability that any given element is contained in both sets is $\frac{m}{k}\frac{n}{k}$.
	By the union bound,
	the probability that these two sets intersect at all is 
	not larger than
	$\frac{mn}{k}$.
	Thus
	\begin{align*}
			&\|(x\star \Id^{\otimes (k-m)})(y\star \Id^{\otimes (k-n)}) - (x \star y \star \Id^{\otimes (k-m-n)})  \|_{\Dcal^\infty}\\ 
			= 
			&\Big\|
			\frac{1}{(k!)^2} \sum_{\stackrel{\pi,\pi' \in S_k}{\pi([m])\cap\pi'([n])\neq \emptyset}} \alpha_\pi(x \otimes \Id^{\otimes (k-m)}) \alpha_{\pi'}(y \otimes \Id^{\otimes (k-n)}) 
			\Big\|_{\Dcal^\infty}\\
		\leq &\frac{mn}k \norm{x}_{\Dcal^\infty} \norm{y}_{\Dcal^\infty}.
	\end{align*}
\end{proof}

\begin{proof}[Proof of Lemma~\ref{lem:Cstarnorm}]
	For $x \in \Sym^n(\Dcal)$, we have the estimate
	\begin{align*}
					\|x\star\Id\|_{\Dcal^\infty}
					=
					\Big\|
					\frac{1}{(n+1)!} \sum_{\pi\in S_{n+1}} \alpha_\pi(x\otimes\Id)
					\Big\|_{\Dcal^\infty}
					\leq
					\frac{1}{(n+1)!} \sum_{\pi \in S_{n+1}} 
					\left\|
					\alpha_\pi(x\otimes \Id)
					\right\|_{\Dcal^\infty}
					=
					\norm{x \otimes \Id}_{\Dcal^\infty}
					\leq
					\|x\|_{\Dcal^\infty}.
	\end{align*}
	Using this estimate repeatedly shows that the sequence 
	$\|x\star \Id^{\otimes k}\|_{\Dcal^\infty}=\|(x\star \Id^{\otimes (k-1)})\star \Id\|_{\Dcal^\infty}$ 
	is non-increasing and hence convergent.
	Subadditivity, absolute homogeneity, and invariance under involution of $\|\cdot\|_{\Sym}$ on $\cup_n \Sym^n(\Dcal)$ follow directly from the same properties of $\|\cdot\|_{\Dcal^\infty}$.
  For $x \in \Sym^m(\Dcal)$ and $y \in \Sym^n(\Dcal)$, Lemma~\ref{lem:combinatorics} implies the $C^*$-norm property
	\begin{align*}
		\begin{split}
			\|x\star y\|_{\Sym} &= 
			\lim_{k\to\infty}
			\|x\star y \star \Id^{\otimes (k-n-m)}\|_{\Dcal^\infty}  \\
			&=
			\lim_{k\to\infty}
			\|(x\star \Id^{\otimes(k-n)})(y \star \Id^{\otimes (k-m)})\|_{\Dcal^\infty}  \\
			&\leq
			\lim_{k\to\infty}
			\|(x\star \Id^{\otimes (k-n)})\|_{\Dcal^\infty}\|(y \star \Id^{\otimes (k-m)})\|_{\Dcal^\infty}  \\
			&=
			\|x\|_{\Sym} 
			\|y\|_{\Sym} 
	  \end{split}
	\end{align*}
	with equality if $y=x^*$.

  We have thus verified the $C^*$-seminorm properties, and the second advertised claim follows from the general theory \cite[Section~6.1]{murphy1990c}.
\end{proof}

Next, we aim to set up a bijection between the space of symmetric states $K_s(\Dcal^\infty)$ and the state space $K(\Sym^\infty(\Dcal))$ of the abelian algebra.
The connection revolves around $\cup_n\Sym^n(\Dcal)$, as it can be interpreted as a subspace of either algebra.
We will thus look for natural ways of extending a state $\rho$ from 
$\cup_n\Sym^n(\Dcal)$ to 
$\Sym^\infty(\Dcal)$ and to 
$\Dcal^\infty$ respectively.

For the former case, we can use the fact that $\cup_n\Sym^n(\Dcal)$ is dense in $\Sym^\infty(\Dcal)$.
Thus, if $\rho\in K_s(\Dcal^\infty)$, it is natural to try to extend 
it 
by continuity from
$\cup_n \Sym^n(\Dcal)$ to a state on all of $\Sym^\infty(\Dcal)$.
Lemma~\ref{lem:states} shows that this ansatz indeed leads to a well-defined map
\begin{align}\label{eqn:sym to abelian}
	E: K_s(\Dcal^\infty) \to K(\Sym^\infty(\Dcal)).
\end{align}

Conversely, in order to evaluate a state $\rho\in K(\Sym^\infty(\Dcal))$ on an element of $x\in\Dcal^\infty$, our approach is to map $x$ to a symmetrized version $\Sym(x)\in\Sym^\infty(\Dcal)$ and then to apply $\rho$ to $\Sym(x)$.
To define the symmetrization operation, 
note that  any element of $\Dcal^\infty$ can be represented by a Cauchy sequence $(x_n)_n$ with $x_n\in \Dcal^n$
and set
\begin{align}\label{eqn:def sym}
	\Sym: (x_n)_n \mapsto (\Sym^n(x_n))_n.
\end{align}
Lemma~\ref{lem:states} establishes that the 
result lies in $\Sym^\infty(\Dcal)$ and that the adjoint
\begin{align}
	(\Sym^*(\rho))(x) = \rho(\Sym(x))
\end{align}
defines a map 
\begin{align}\label{eqn:abelian to sym}
	\Sym^*: K(\Sym^\infty(\Dcal)) \to K_s(\Dcal).
\end{align}

\begin{lemma}\label{lem:states}
	The maps $\Sym^*$ and $E$
	are well-defined and inverses of each other.
	What is more, $\Sym^*$ is weakly continuous.
\end{lemma}

\begin{proof}
	We will repeatedly make use of the fact
  \cite[Prop.~II.6.2.5]{blackadar2006operator}
	that the states on a $C^*$-algebra are exactly those functionals $\rho$ that satisfy 
	\begin{align}\label{eq:characterization}
			\rho(\Id) = 1,\qquad
			\abs{\rho(x)} \leq \|x\|.
	\end{align}

	Eq.~(\ref{eqn:def sym}) indeed defines a map from $\Dcal^\infty\to\Sym^\infty(\Dcal)$:
	If $(x_n)_n, x_n\in\Dcal^n$ is a Cauchy sequence with respect to $\|\cdot\|_{\Dcal^\infty}$, then by Eq.~(\ref{eqn:norm contraction}), the sequence $(\Sym^n(x_n))_n$ 
	is Cauchy with respect to $\|\cdot\|_{\Sym}$ and therefore an element of $\Sym^\infty(\Dcal)$.
	Next, let $\rho\in K(\Sym^\infty(\Dcal))$.
	Then
	\begin{align*}
		\rho(\Sym(\Id)) = \rho(\Id) = 1,
		\qquad
		|\rho(\Sym(x))| \leq \|\Sym(x)\|_{\Sym} \leq \|x\|_{\Dcal^\infty},
	\end{align*}
	thus $\Sym^*(\rho)$ is a state.
	Because $\Sym \circ \alpha_\pi = \Sym$ for any permutation $\pi$, $\Sym^*(\rho)$ is symmetric.
	The map $\Sym^*$ is weakly continuous:
	If a net 
	$\rho_\lambda$ 
	in $K(\Sym^\infty(\Dcal))$
	converges weakly to $\rho$,
	then in particular $\rho_\lambda(\Sym(x)) \to \rho(\Sym(x))$ for all $x\in\Dcal$.
	Thus $\Sym^*(\rho_\lambda)$ converges weakly to $\Sym^*(\rho)$.

	To prove that $E$ is well-defined, start with a state $\rho\in K_s(\Dcal^\infty)$.
	For $x\in \cup_n \Sym^n(\Dcal)$,
	using symmetry and Eq.~(\ref{eq:characterization}), 
	\begin{align*}
		|\rho(x)|
		=
		\lim_{k\to\infty} |\rho(x\star \Id^{\otimes k})|
		\leq
		\lim_{k\to\infty} \|x\star \Id^{\otimes k}\|_{\Dcal^\infty}
		=\|x\|_{\Sym}.
	\end{align*}
	In other words, on $\cup_n\Sym^n(\Dcal)$,  $\rho$ is bounded with respect to the $\|\cdot\|_{\Sym}$-norm and can thus be uniquely extended by continuity to a functional $E(\rho)$ on $\Sym^\infty(\Dcal)$.
	Using Eq.~(\ref{eq:characterization}) once more, the preceding estimate also shows that $E(\rho)$ is a state.

	Finally, for each $\rho\in K_s(\Dcal^\infty), x\in\Dcal^\infty$ and $\sigma\in K(\Sym^\infty(\Dcal)), y\in\cup_n\Sym^n(\Dcal)$,
	\begin{align*}
		\Sym^*(E(\rho))(x) = E(\rho)(\Sym(x))=\rho(\Sym(x)) = \rho(x),
		\qquad
		E(\Sym^*(\sigma))(y) = (\Sym^*\sigma)(y) = \sigma(y),
	\end{align*}
	which shows that the two maps are inverses of each other.
\end{proof}

\begin{proof}[Proof of Theorem~\ref{thm:finetti maximale}]
	Consider $E(\rho)\in K(\Sym^\infty(\Dcal))$.
	By \cite[Example~4.1.30 and Proposition~2.3.27]{bratteli2012operator}, because $\Sym^\infty(\Dcal)$ is abelian, there exists a unique measure $\tilde\mu$ 
	over pure states $K_{\text{pure}}(\Sym^\infty(\Dcal))$,
	such that 
	\begin{align*}
		E(\rho)(x)
		=
		\int_{K_{\text{pure}}(\Sym^\infty (\Dcal))}
		\tilde\sigma(x)\,\mathrm{d}\tilde\mu(\tilde\sigma).
	\end{align*}
	But then, for $x_i\in\Dcal$,
	\begin{align*}
		\rho(x_1\otimes\ldots\otimes x_n)
		=
		\rho(x_1\star \ldots\star x_n) 
		=
		(E(\rho))(x_1\star \ldots\star x_n) 
		=
		\int_{K_{\text{pure}}(\Sym^\infty (\Dcal))}
		\tilde\sigma(x_1\star \ldots\star x_n) 
		\,\mathrm{d}\tilde\mu(\tilde\sigma). 
	\end{align*}
	Consider one $\tilde\sigma\in K_{\text{pure}}(\Sym^\infty(\Dcal))$,
	let $R:K(\Sym^\infty(\Dcal))\to K(\Dcal)$ be the map that restricts states 
	to $\Dcal\subset\Sym^\infty(\Dcal)$,
	and let $\sigma = R(\tilde\sigma)$.
	Because pure states of abelian algebras are homomorphisms,
	\begin{align*}
		\tilde\sigma(x_1\star \ldots\star x_n) 
		=
		\tilde\sigma(x_1)\dots \tilde\sigma(x_n)
		=
		\sigma(x_1)\dots \sigma(x_n)
		=
		\Pi_\sigma(x_1\otimes\dots\otimes x_n).
	\end{align*}

	The restriction $R$ is the adjoint of the embedding $\Dcal\to\Sym^\infty(\Dcal)$.
	As the adjoint of a bounded map, it is weak$^*$-continuous by the same argument as the one used in the proof of Lemma~\ref{lem:states}.
	Thus the concatenation $R \circ \Sym^*: \tilde\sigma\mapsto \sigma$ is continuous and hence measurable.
	We can therefore define a measure $\mu$ on $K(\Dcal)$ by 
	\begin{align*}
		\mu(S) = \tilde\mu\big( (R \circ \Sym^*)^{-1}(S) \big).
	\end{align*}
	Then
	\begin{align*}
		\rho(x_1\otimes\ldots\otimes x_n)
		=
		\int_{K(\Dcal)}
		\Pi_\sigma(x_1\otimes\ldots\otimes x_n)
		\,\mathrm{d}\mu(\sigma),
	\end{align*}
	which proves the claim, as $\cup_n\Dcal^n$ is dense in $\Dcal^\infty$.
\end{proof}

\section{A convergent hierarchy} \label{sec:new_SDP}

Motivated by the difficulties that were outlined in Section~\ref{sec:background_challenges},
we propose a modified hierarchy of semidefinite programs for a Schmidt rank-constrained version of the quantum causal optimization problem that is provably complete. 
We show that by increasing the Schmidt rank, one can approximate any POVM arbitrarily well. 

We use the triangle causal structure without settings as a guiding example to demonstrate the technique and to keep the notation relatively legible.
More general scenarios can be accommodated -- e.g.\ it is straight-forward to add additional generators to describe several possible POVMs per party.
Extensions of these methods to arbitrary quantum causal structures are discussed in Section~\ref{sec:arbitrary_quantum_causal_structures}.

\subsection{Construction of the hierarchy}

\subsubsection{The universal algebra of the quantum causal structure}

First, we define generators and relations for the universal $C^*$-algebra $\Dcal^n$ modeling the most general set of observables for the $n$-th inflation level of the causal structure.
The algebra depends on a number of parameters:
\begin{enumerate}
	\item
The causal structure (taken to be the triangle scenario for now);
	\item
The number of outcomes $M$ per vertex; 
	\item
The inflation level $n$;
	\item
A bound $r$ on the Schmidt rank of the measurement operators;
	\item
A bound $C$ on the norm of the generators of the local algebra.
\end{enumerate}
The dependency of $\Dcal^n$ on the parameters will not be made explicit, with the exception of the inflation level.

From this data, define the set $\Gcal^n$ of $6(M-1)rn+1$ generators to be
\begin{align*}
	\{\Id\}
	\cup
	\big\{
	e_-^i(a,\alpha), e_+^i(a,\alpha),
	f_-^i(a,\alpha), f_+^i(a,\alpha),
	g_-^i(a,\alpha), g_+^i(a,\alpha)\>|\>
	a\in\{1, \dots, M-1\}, 
	\alpha \in\{ 1, \dots, r\},
	i \in\{ 1, \dots, n\}
	\big\}. 
\end{align*}
We will use the abbreviations
\begin{align} \label{eq:split_Ea}
	E_a^{ij} := \sum_{\alpha=1}^r e^i_-(a,\alpha) e^j_+(a,\alpha), \quad
	F_a^{ij} := \sum_{\alpha=1}^r f^i_-(a,\alpha) f^j_+(a,\alpha), \quad
	G_a^{ij} := \sum_{\alpha=1}^r g^i_-(a,\alpha) g^j_+(a,\alpha)
\end{align}
and
\begin{align*}
	X_M^{ij} := \Id - \sum_{a=1}^{M-1} X_a^{ij},
	\qquad
	X\in\{E, F, G\}.
\end{align*}
Four types of constraints are imposed.
Locality constraints: 
\begin{align}
	\text{for all }
	&x\in\{e,f,g\},\,
	y\in\{-, +\},\, 
	a\in\{1, \dots, M\},\,
	\alpha\in\{1, \dots, r\},\,
	i \in\{ 1, \dots, n\}: \nonumber \\
	&[x_y^i(a,\alpha), x'^{i'}_{y'}(a',\alpha')]=0
	\qquad \text{ unless } x=x', y=y', \text{ and } i=i'.  \label{eq:locality constraints}
\end{align}
Measurement constraints: 
\begin{align} \label{eq:proj_meas_constr}
	X_a^{ij} \quad \text{ is positive},
	\qquad
	X \in \{ E, F, G \}.
\end{align}
Norm constraints:
\begin{align} \label{eq:bound}
	\|g\|\leq C, \qquad \forall\>\Id \neq g \in \Gcal^n.
\end{align}
And finally that 
\begin{align*}
	\Id x = x \Id = x
\end{align*}
for all generators $x$.
Together, these constraints define the set of relations $\Rcal^n$.
The NPO will run over states on the universal $C^*$-algebra $\Dcal^n=C^*(\Gcal^n|\Rcal^n)$.

\subsubsection{Polynomial constraints and objective function}

The quantum causal polynomial optimization problem minimizes a polynomial function $f_0$ over compatible states  $\rho\in K(\Dcal^1)$ that also fulfill a number of polynomial constraints $f_i(\rho)=0$.
Here, we construct these objects precisely.

Choose some $g, k\in \NN$.
Recall the definition of the finite vector space $\Fcal^{(k)}(\Gcal)$ of polynomials of order $k$ in the generators $\Gcal$ from Sec.~\ref{subsubsec:NPO}.
We assume that the functions are such that for every $f_i$, there exists a $y_i$ in the $g$-fold algebraic tensor product $\mathcal{F}^{(k)}\otimes_\mathrm{alg}\dots\otimes_\mathrm{alg}\Fcal^{(k)}$ such that $f_i(\rho)$ equals the evaluation of the product state $\rho^{\otimes g}$ on $y_i$: 
\begin{align}\label{eqn:polarization}
	f_i(\rho) = \rho^{\otimes g}(y_i)
\end{align}
For our purposes, it will be enough to take Eq.~(\ref{eqn:polarization}) as the definition of the type of functions we allow for.
We remark, though, that passing from a degree-$g$ polynomial function $f_i$ on $\mathcal{F}^{(k)}$ to a $y_i\in(\mathcal{F}^{(k)})^{\otimes_\mathrm{alg} g}$ such that Eq.~(\ref{eqn:polarization}) holds is known as a \emph{polarization} in multi-linear algebra. 
In this context, it is proven that a unique suitable $y_i$ always exists.

As an example, consider the 2-norm distance that allows one to reduce Problem~\ref{prob:compatibility} to Problem~\ref{prob:optimization} as we will see in Corollary~\ref{cor:causal_compatibility}. The objective function is then given by
\begin{align} \label{eq:causal_comp_opt}
    \sum_{a,b,c} \left( \rho(E^{11}_a F^{11}_b G^{11}_c ) - P(a,b,c) \right)^2.
\end{align}
To find the polarization, note that for a compatible state $\rho$ it holds that
\begin{align}
    &\sum_{a,b,c} \left( \rho(E^{11}_a F^{11}_b G^{11}_c ) - P(a,b,c) \right)^2\nonumber \\
    =&\sum_{a,b,c} \rho(E^{11}_a F^{11}_b G^{11}_c )\rho(E^{22}_a F^{22}_b G^{22}_c ) -2P(a,b,c) \rho(E^{11}_a F^{11}_b G^{11}_c ) + P(a,b,c)^2\nonumber \\
		=&\rho^{\otimes 2}\Big(\sum_{a,b,c} E^{11}_a F^{11}_b G^{11}_c E^{22}_a F^{22}_b G^{22}_c -2P(a,b,c) E^{11}_a F^{11}_b G^{11}_c  + P(a,b,c)^2\Big), \label{eq:polarization_example}
\end{align}
which is indeed of the form $\rho^{\otimes g}(y_i)$.

We have now assigned a precise meaning to every object that appeared in the quantum causal polynomial optimization problem (Problem~\ref{prob:optimization}), which we restate here with constraints on the Schmidt rank of the POVM elements and the norm of the generators (i.e.~as in Problem~\ref{prob:optimization rank}):
	Given a causal structure, 
	a choice for the parameters $M, r, C$,
	and a family of polynomial functions $f_i$ on $K(\Dcal^1)$ as defined above and such that the
	$f_i, i\geq 1$ are non-negative on states that are compatible with the causal structure.
	Find
	\begin{align}
\begin{split} \label{eq:opt_rank}
		f_{r,C}^\star &= \min_{\rho\in K(\Dcal^1)} f_0(\rho)  \\
			\text{s.\ t.} \quad &f_i(\rho) = 0 \quad i\geq 1 \\
			&\rho \text{ \rm{is compatible with the causal structure}}.
	\end{split}
	\end{align}

We adopt the common convention that $f^\star_{r,C}$ is $\infty$ in case the problem is infeasible.

\subsubsection{NPO formulation}

We now pass to an NPO problem, which we will show is asymptotically equivalent to the causal optimization problem in Eq.~(\ref{eq:opt_rank}).
To do so, we will replace the polynomial functions $f_i$ by their polarizations $y_i$, and replace the causal constraint on $\rho$ by symmetry constraints on a degree-$n$ inflation.

Choose some $n$ larger than or equal to the degree of $y_0$.
For permutations $\pi, \pi', \pi''\in S_n$ define
an action on generators:
\begin{align}
\begin{split}
	e_+^i \mapsto e_+^{\pi(i)},  \quad f_-^i &\mapsto f_-^{\pi(i)}, \\
	f_+^i \mapsto f_+^{\pi'(i)}, \quad g_-^i &\mapsto g_-^{\pi'(i)}, \\
	g_+^i \mapsto g_+^{\pi''(i)}, \quad e_-^i &\mapsto e_-^{\pi''(i)}.
\end{split}
\end{align}
Let $\alpha_{\pi,\pi',\pi''}$ be the extension of this action to $\Fcal(\Gcal^n)$.

The NPO problem relaxation of (\ref{eq:opt_rank}) at inflation level $n$ is
\begin{align}
\begin{split} \label{eq:optimization_new}
	f_{r,C}^n &= \min_{\rho\in K(\Dcal^n)} \rho(y_0)  \\
	\text{s. t.} \quad &\rho(y_i) = 0,\qquad   i\geq 1, y_i \in \Fcal^{(n)}(\Gcal), \\
										 &\rho(b_j^{(n)} - \alpha_{\pi,\pi',\pi''}(b_j^{(n)})) = 0, 
	\end{split}
	\end{align}
	where the final symmetry constraint ranges over a basis $\{b_j^{(n)}\}$ for $\Fcal^{(n)}(\Gcal)$ and a generating set of permutations in $S_n^{\times 3}$.
	
We note that (\ref{eq:optimization_new}) is not yet directly a semi-definite program.
Instead, every instance gives rise to the infinite (but complete) hierarchy of SDP relaxations discussed in Section~\ref{sec:cstar}.

\subsection{Proof of completeness} \label{sec:completeness}

\begin{figure}
\includegraphics[width=0.3\textwidth]{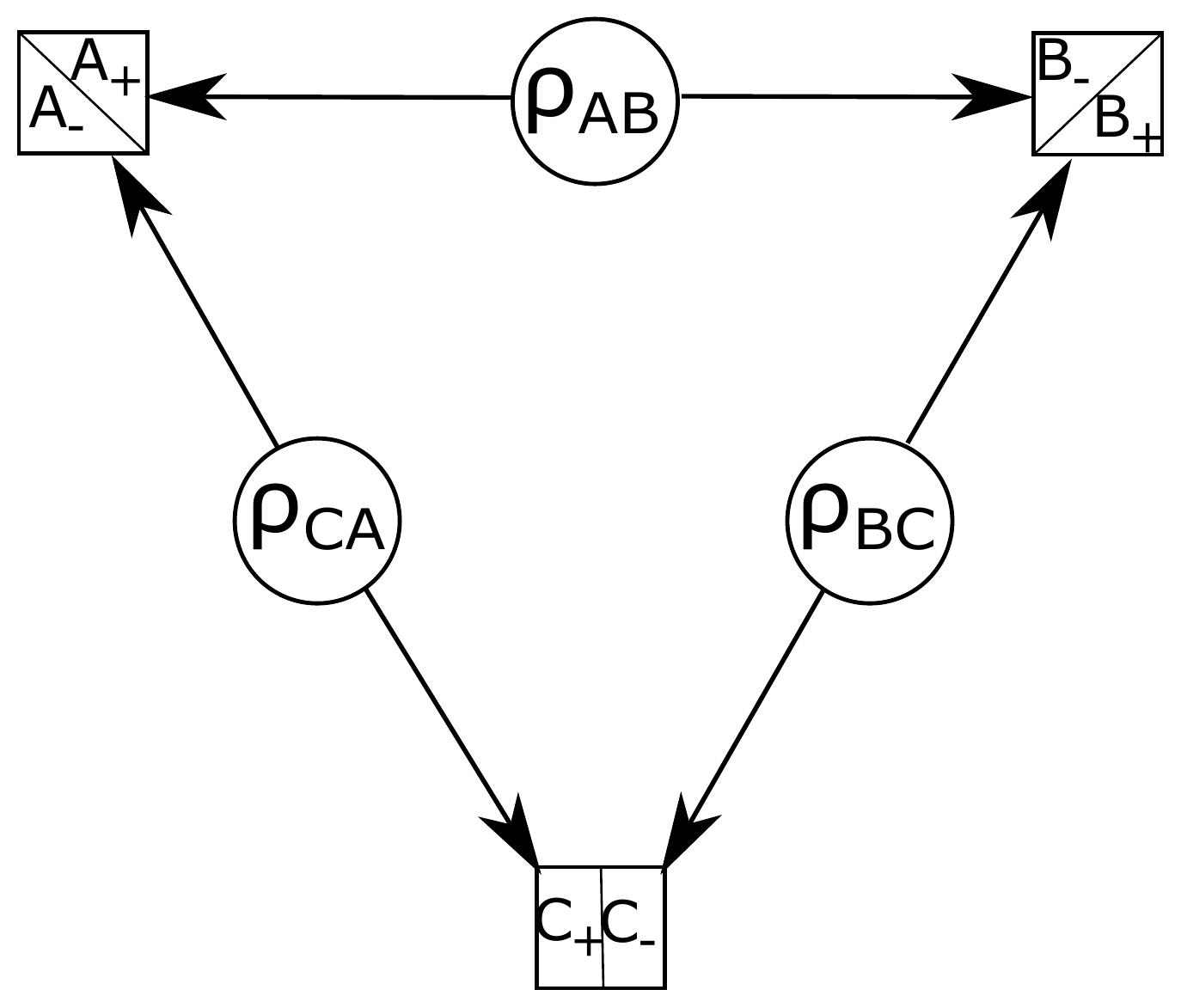}
\caption{We associate to each quantum system an observable algebra. This means that the POVM elements of Alice generate a subalgebra of a larger algebra $\Acal_- \otimes \Acal_+$, and similar for Bob and Charlie. It is with respect to the splitting $\Ccal_+ \Acal_- \mid \Acal_+ \Bcal_- \mid \Bcal_+ \Ccal_-$ that the global state is supposed to factorize.}
\label{fig:triangle_split}
\end{figure}

Now follows the proof that the inflation hierarchy of Eq.~\eqref{eq:optimization_new} is complete, i.e.\ that in the limit of $n\to\infty$, Eq.~\eqref{eq:optimization_new} and Eq.~\eqref{eq:opt_rank} are equivalent. 
Afterwards, we show that for every $\epsilon$ in the approximate quantum causal compatibility problem (Problem \ref{prob:compatibility}), there exist a Schmidt rank $r$ and a norm bound $C$ such that any compatible distribution can be $\epsilon$-approximated by one that can be realized in a model that respects the bounds on $r, C$.

\begin{theorem}\label{thm:hierarchy_convergence}
	The hierarchy \eqref{eq:optimization_new} is complete for the problem \eqref{eq:opt_rank} in the sense that
	\begin{align*}
		f_{r,C}^{\infty} := \lim_{n} f_{r,C}^n = f_{r,C}^\star.
	\end{align*}
\end{theorem}

\begin{proof}
	Since each level of the hierarchy is a relaxation of the original problem, it holds that
	\begin{align} \label{eq:ineq_easy_way}
			f_{r,C}^n \leq f_{r,C}^\star \qquad \forall n.
	\end{align}

	The converse inequality is more involved. 
	We start by constructing a state $\omega_n$ on $\Dcal^\infty$ for each $n$ by taking the infinite tensor product of some optimizing state of the problem in Eq.~(\ref{eq:optimization_new}).
	By the Banach-Alaoglu Theorem applied to the state space $K(\Dcal^\infty)$, there exists a  weak$^*$-convergent subsequence of the $\omega_n$.
	Let $\omega$ be its limit point.

	For each $i\geq 1$,
	$y_i$ has a finite degree $n_i$. 
	The constraint $\rho(y_i)=0$ 
  in \eqref{eq:optimization_new}
	implies that $\omega_n(y_i)=0$ for every $n\geq n_i$, and therefore the same is true for the limit:
	\begin{align*}
		\omega(y_i)=0.
	\end{align*}
	Because $\omega_n$ is chosen to be an optimizer, $\omega_n(y_0)=f^n_{r,C}$ and thus 
	\begin{align*}
		\omega(y_0) = \lim_{n\to\infty} f^n_{r,C} = f^{\infty}_{r,C}.
	\end{align*}
	Likewise, the symmetry constraints 
  in \eqref{eq:optimization_new}
	imply that 
	\begin{align}\label{eqn:symmetry limit}
		\omega \circ \alpha_{\pi,\pi',\pi''} = \omega.
	\end{align}
  Restricting to the diagonal case $\pi=\pi'=\pi''$, we conclude that
	the limit $\omega$ is a symmetric state on $\Dcal^\infty$, so that
	Theorem~\ref{thm:finetti maximale} applies.
	
	Next, for each $1\leq n \leq \infty$, introduce the algebras
	\begin{align*}
		(\Ccal_+ \Acal_-)^{n}, \qquad (\Acal_+ \Bcal_-)^{n},	 \qquad (\Bcal_+ \Ccal_-)^{n},
	\end{align*}
	where $(\Ccal_+ \Acal_-)^n \subset \Dcal^n$ is the subalgebra 
	generated by 
	$\bigcup_{i\leq n}\bigcup_{a,\alpha} \{ g_+^i(a,\alpha), e_-^i(a,\alpha) \}$,
	and similar for $(\Acal_+ \Bcal_-)^n$ and $(\Bcal_+ \Ccal_-)^n$.	
	As $n$ ranges over all natural numbers, the linear span of elements of the form
	\begin{align*}
		x = u v w,
	\end{align*}
	with $u \in (\Ccal_+ \Acal_-)^{n}$, $v \in (\Acal_+ \Bcal_-)^{n}$ and $w \in (\Bcal_+ \Ccal_-)^{n}$ is dense in $\Dcal^\infty$.
	We therefore lose no generality by restricting the analysis of the action of $\omega$ to elements of this form.

	Fix one $n\in\NN$ and $x=uvw\in\Dcal^n$.
	Using the cycle notation, define the permutations 
	\begin{align*}
		\pi &= (1,n+1) \, (2, n+2)\, \dots\, (n,2n)  &&\text{(i.e.\ exchange the 1st block of $n$ symbols with the 2nd block of $n$ symbols)},\\
		\pi' &= (1,2n+1) \, (2, 2n+2)\, \dots\, (n,3n)  &&
		\text{(i.e.\ exchange the 1st block of $n$ symbols with the 3rd block of $n$ symbols)}.
	\end{align*}
	Then
	\begin{align}
			\omega(x) &= \omega(\alpha_{\Id,\pi,\pi'}(x))\\
								&= 	\omega \big(
											u \, \alpha_\pi(v)\, \alpha_{\pi'}(w) 
										\big) \label{eqn:non diag symmetry}\\ 
								&= \int \dd \mu(\sigma)\ {\Pi_\sigma} \big( 
											u \, \alpha_\pi(v)\, \alpha_{\pi'}(w)\, 
								\big) 
			\label{eq:omega_separable}\\
								&= 
								\int \dd \mu(\sigma)\ 
								{\Pi_\sigma} ( u ) \,
								{\Pi_\sigma} \big( \alpha_\pi(v) \big) \,
								{\Pi_\sigma} \big( \alpha_{\pi'}(w) \big)
			\label{eq:omega_prod} \\
			&= 
								\int \dd \mu(\sigma)\ 
								{\Pi_\sigma} ( u ) \,
								{\Pi_\sigma} ( v ) \,
								{\Pi_\sigma} ( w ),
			\label{eq:omega_symmetric}
	\end{align}
	where 
	Eq.~\eqref{eqn:non diag symmetry} follows from Eq.~\eqref{eqn:symmetry limit},
	Eq.~\eqref{eq:omega_separable} from Theorem~\ref{thm:finetti maximale}, and in Eqs.~\eqref{eq:omega_prod} and \eqref{eq:omega_symmetric} we have used that $\Pi_\sigma$ is a symmetric product state for disjoint sets of layers of the inflation. 

	For each $\sigma$, the integrand in Eq.~(\ref{eq:omega_symmetric}) factorizes (c.f.~Fig.~\ref{fig:triangle_split}).
	The respective marginals of $\Pi_\sigma$ will be denoted as
	\begin{align}
		\Lambda_{\sigma}^{\Ccal_+ \Acal_-} := \Pi_{\sigma} |_{(\Ccal_+  \Acal_-)^\infty}, &
		\qquad \Lambda_{\sigma}^{\Acal_+ \Bcal_-} := \Pi_{\sigma} |_{(\Acal_+  \Bcal_-)^\infty}, 
		\qquad \Lambda_{\sigma}^{\Bcal_+ \Ccal_-} := \Pi_{\sigma} |_{(\Bcal_+  \Ccal_-)^\infty},
	\end{align}
	so that the product state appearing in the integrand is
	\begin{align}
		&\Lambda_{\sigma} := \Lambda_{\sigma}^{\Ccal_+ \Acal_-} \otimes \Lambda_{\sigma}^{\Acal_+ \Bcal_-} \otimes \Lambda_{\sigma}^{\Bcal_+ \Ccal_-}.
	\end{align}
	Therefore, $\omega$ is a convex combination 
	\begin{align}
		\omega(x) &= 
		\int \dd \mu(\sigma)\ \Lambda_{\sigma} (x).
	\end{align}
of states $\Lambda_\sigma$ that are compatible with 
	the causal structure.

	It remains to be shown that we can choose one $\sigma$, such that $\Lambda_\sigma(y_0) = f^\infty_{r,C}$ and $\Lambda_\sigma(y_i) = 0$ for all $i \geq 1$.

	By the definition of Problem~\ref{prob:optimization}, the $y_i$ are non-negative on states compatible with the causal structure, i.e.\ $\Lambda_\sigma(y_i)\geq 0$ for all $i\geq 1$.
	Because $\omega(y_i)=0$ as well, the constraints must be fulfilled on a set $E\subset K(\Dcal)$ of measure $\mu(E)$ equal to one.
	For every $\Lambda_\sigma$ with $\sigma \in E$ it must hold that $\Lambda_\sigma(y_0) \geq f^\infty_{r,C}$, for else one could have chosen $\mu$ to be the point measure on a state $\sigma' \in E$ with $\Lambda_{\sigma'}(y_0) < f^\infty_{r,C}$, which contradicts the fact that $f^\infty_{r,C}$ is a minimum.
	As before, there must be a subset $F\subset E$ of measure $\mu(F)$ equal to one such that $\Lambda_\sigma(y_0)=f^\infty_{r,C}$ for all $\sigma\in F$.
	Therefore, any state $\Lambda_\sigma$ such that $\sigma \in F$ is compatible with the constraints of Eq.~\eqref{eq:opt_rank}, so that we can conclude
	\begin{align} \label{eq:ineq_hard_way}
		f^\infty_{r,C} 
					= 
					\Lambda_{\sigma}(y_0) \geq f^\star_{r,C} \quad \forall \sigma \in F.
	\end{align}
	Combining Eqs.~\eqref{eq:ineq_easy_way} and \eqref{eq:ineq_hard_way} yields $f^\infty_{r,C} = f^\star_{r,C}$.
\end{proof}

\vspace{0.2cm}
\noindent \textbf{Remarks.} 
\begin{enumerate}
	\item
It is not obviously possible to extract a compatible state from the SDP.
	\item
	The proof of Theorem \ref{thm:hierarchy_convergence} shows that it is in general not possible to add additional constraints of the form $\rho(x) \geq 0$ to the program for elements $x$ that are not necessarily positive on compatible states: the states $\Lambda_\sigma$ that are compatible with the causal structure might not obey these constraints, since they only apply to $\omega$. 
That is, the set $E$ defined in the proof  will in general not have full measure.
However, if the optimization problem is a feasibility problem, i.e.~if it has a trivial objective function, it is possible to put one such constraint as the objective function and reject the solution if the optimal value does not obey the inequality. 
We will apply this to the constraints of the causal compatibility problem in Corollary \ref{cor:causal_compatibility} below.
\end{enumerate}

\begin{lemma} \label{lem:two-norm}
Consider a probability distribution $P$ that is compatible with a given causal structure. 
Choose $\epsilon>0$.
There exist constants $C, r$ such that
there is a distribution $\tilde P$ that approximates $P$ in the sense that
\begin{align*}
		\|P-\tilde P\|_2^2 \leq \epsilon
\end{align*}
which can be realized using only POVM elements of the form
\begin{align} \label{eq:POVM_decomp}
	\tilde E = 
	\sum_{\alpha=1}^r e_-(\alpha) \cdot e_+(\alpha), \qquad \text{\rm such that  }\|e_-(\alpha)\|, \|e_+(\alpha)\| \leq C.
\end{align}
\end{lemma}

\begin{proof}
	Consider the original model that gives rise to $P$.
	By the definition of a $C^*$-tensor product, for each $a=1, \dots, M-1$, there is a convergent series
	\begin{align*}
		E_a=
		\sum_{\alpha=1}^\infty e_-(a,\alpha) \cdot e_+(a,\alpha),
		\qquad e_-(a,\alpha)\in \mathcal{A}_-, e_+(a,\alpha) \in \Acal_+.
	\end{align*}
	Let $E_a^{(r)}$ be the truncation of the series to the first $r$ terms.
	Convergence implies that for every $\delta>0$, there exists an $r$ such that 
	\begin{align*}
		\|E^{(r)}_a - E\| \leq \delta\qquad a=1, \dots, M-1.
	\end{align*}
	
	What remains to be proven is that one can turn these partial sums into an exact POVM.
  To this end, set
	\begin{align}\label{eqn:approx povm}
		\tilde E_a = 
		\frac1{1+2M\delta}
		(\delta \Id + E^{(r)}_a)\qquad a = 1, \dots, M-1.
	\end{align}
	Then the $\tilde E_a$ are positive.
	What is more,
	\begin{align*}
		\left\| 
	  	\sum_{a=1}^{M-1} \tilde E_a
		\right\|
		&\leq
		\frac{1}{1+2M\delta}
		\left(
			\sum_a^{M-1} \|\delta\Id\|
			+
			\left\|
			\sum_a^{M-1} (E_a^{(r)}+E_a - E_a)
			\right\|
		\right) \\
		&\leq
		\frac{1}{1+2M\delta}
		\left( 
			(M-1) \delta
			+
			\left\|
			  \sum_a^{M-1} E_a
			\right\|
			+
			\left\|
			\sum_a^{M-1} (E_a^{(r)}-E_a)
			\right\|
		\right) \\
		&\leq
		\frac{1}{1+2M\delta}
		\big(
			(M-1)\delta
			+
			1+ (M-1)\delta
		\big)
		< 1
	\end{align*}
	so that
	\begin{align*}
		\tilde E_{M} := \Id - \sum_{a=1}^{M-1} \tilde E_a
	\end{align*}
	is also positive.
	Therefore $\{\tilde E_1, \dots, \tilde E_M\}$ forms a POVM.
	Repeating the construction, one arrives at approximations $\tilde F_b$ to $F_b$ and $\tilde G_c$ to $G_c$.

	From Eq.~(\ref{eqn:approx povm}), the approximating POVM elements converge to the original ones in operator norm as $\delta\to 0$.
	The same is thus true for all polynomial expressions in the POVM elements.
	Therefore,
	\begin{align*}
		\tilde P(a,b,c):=\rho(\tilde E_a \tilde F_b \tilde G_c) \to P(a,b,c) \qquad (\delta\to 0)
	\end{align*}
	and, because there are only finitely many outcomes,
	\begin{align*}
		\|\tilde P - P\|_2^2 \to 0\qquad (\delta\to 0).
	\end{align*}
	Thus, choosing $r$ sufficiently high, an arbitrarily good approximation can be achieved.
	The advertised claim follows by choosing $C$ to be the largest operator norm of any factor of the partial sums involved.

\end{proof}

\begin{corollary} \label{cor:causal_compatibility}
Given a probability distribution $P$ over observed variables, the SDP hierarchy that corresponds to the optimization problem of Eq.~\eqref{eq:optimization_new} can solve the approximate quantum causal compatibility problem described in Problem \ref{prob:compatibility}.
\end{corollary}

\begin{proof}
In order to solve Problem \ref{prob:compatibility}, we need to show that there exists a state $\rho$ that is compatible with the description of the causal structure and that produces statistics that are close in 2-norm to the observed statistics. In particular, for the triangle scenario it must hold that
\begin{align} \label{eq:norm_distance}
    \sum_{a,b,c} \left( \rho(E^{11}_a F^{11}_b G^{11}_c ) - P(a,b,c) \right)^2 \leq \epsilon,
\end{align}
where $\rho$ is a compatible state.
Once again we can polarize the expression, which yields the objective function
\begin{align} \label{eq:causal_compatibility_obj}
    \min_{\rho \in K(\Dcal^2)} \sum_{a,b,c} \rho(E^{11}_a F^{11}_b G^{11}_c E^{22}_a F^{22}_b G^{22}_c ) -2P(a,b,c) \rho(E^{11}_a F^{11}_b G^{11}_c ) + P(a,b,c)^2,
\end{align}
as was shown in Eq.~\eqref{eq:polarization_example}.

If the NPO hierarchy attains the optimal value $f_{r,C}^\infty$ for this objective function, there also exists a product state $\Pi_\sigma \in K(\Dcal^\infty)$ that is compatible with the infinitely inflated causal structure that attains the same optimal value by Theorem \ref{thm:hierarchy_convergence}. 
If $f_{r,C}^n > \epsilon$ for any $n$, we reject the hypothesis that the given description of the causal structure with measurement operators of rank $r$ and generators with a norm-bound of $C$ can produce the observed statistics. If there does exist a quantum description of $P(A,B,C)$ Lemma \ref{lem:two-norm} ensures that there exist $r$ and $C$ such that the optimal value is not rejected for any $n$.
In that case the restriction of $\Pi_\sigma$ to $\Dcal$ is a product state that is compatible with the triangle causal structure and that approximately produces the probability distribution $P(A,B,C)$.
\end{proof}

\noindent {\bf Remark.} Though in the limit of $n\to \infty$ the objective function \eqref{eq:causal_compatibility_obj} is equivalent to Eq.~\eqref{eq:norm_distance}, Eq.~\ref{eq:causal_compatibility_obj} is likely to be impractical to detect incompatibility. For low values of $n$ the state will not be separable and this objective function can become negative. Therefore, in practice it is more convenient to optimize over the quadratic function \eqref{eq:norm_distance}. The problem is then a convex quadratic program, which can still be solved via semidefinite programming (see e.g.~Ref.~\cite{wolkowicz2012handbook}). Alternatively, one can add the requirement that the measurements evaluate to the correct probabilities as a set of linear constraints on the state. The objective function can then be made trivial, or the problem can be reformulated as an eigenvalue optimization problem (see e.g.~Ref.~\cite[Sec.~5]{pozas2019quantum}). However, as mentioned in the remark below Theorem \ref{thm:hierarchy_convergence}, it becomes unclear whether such a hierarchy is still complete.

\subsection{Arbitrary quantum causal structures} \label{sec:arbitrary_quantum_causal_structures}
In this section we will generalize the results presented above to arbitrary causal structures for which all leaf nodes are observed classical variables. (Recall that one can also define quantum causal structures that give rise to quantum states rather than classical variables \cite{chaves2015information}, but such scenarios are beyond the scope of this paper).
Ref.~\cite[Sec.~V]{wolfe2021quantum} discusses a number of transformation rules that bring such general causal structures into a form amenable to the hierarchy introduced above.
It is argued there that these rules are sufficient to treat every structure that gives rise to classical random variables.
We will not repeat this argument here. 
However, for each of their transformation rules, we will explain how they can be applied in our modified framework.

We will first argue that the approach described in the previous sections is applicable to all \emph{network scenarios}, which are two-layered causal structures with observed leaf nodes, i.e.\ nodes without children.
In the second step, we will map any \emph{latent exogenous} causal structure to such a network scenario. Latent exogenous causal structures are those in which the only unobservable nodes are root nodes, i.e.\ nodes that have no parents.
The final step is to map non-exogenous causal structures to exogenous ones by introducing a new type of node that can also be treated in our model. 

\subsubsection{Network scenarios}
It is not difficult to see that Theorem \ref{thm:hierarchy_convergence} and Corollary \ref{cor:causal_compatibility} can be extended to arbitrary \emph{correlation scenarios}, which are the causal structures that only have a bottom layer of independent latent systems with arrows pointing to a top layer of observed variables. The triangle scenario is an example of a correlation scenario. 
If the causal structure has $L$ latent (quantum) variables, one can employ the proof strategy of Theorem~\ref{thm:hierarchy_convergence} with $L\cdot g$ levels of inflation. Indeed, for the triangle scenario one has 3 latent quantum systems, and thus $3g$ levels of inflation were sufficient. 
The proof that causal polynomial optimization can be solved with an SDP hierarchy as described above remains nearly the same, with the only difference that the algebra $\Dcal^n$ modeling the level-$n$ inflated causal structure has to be defined in accordance with the proposed causal structure. 
To show that causal compatibility is also solved, one just needs to write down a similar objective function as in Eq.~\eqref{eq:causal_compatibility_obj} for the given probability distribution.

By allowing classical root nodes that only have one child in two-layered causal structures, one obtains so called \emph{network scenarios}. The Bell scenario from Fig.~\ref{fig:Bell_scenario} is an example of such a network scenario. 
Whenever a classical, observed variable is an input to another observed variable, we give the POVM elements of the latter variable an extra index. For example in the Bell scenario, where Alice has the input variable $X$, the POVM elements of Alice become $\{E_a^x\}$, with $\sum_a E^x_a = \Id$ for every $x$, so that $x$ can be interpreted as a measurement setting. POVM elements with different measurement settings, e.g.\ $E^x_a$ and $E^{x'}_a$, need no longer commute. 
Though this description introduces more variables to the model, nothing major changes in the proofs.

\subsubsection{Latent exogenous causal structures}
It is also possible to extend our result to quantum causal structures with more than two layers.
The reduction of the general case to the proof methods considered here is not immediate.
We follow Wolfe \emph{et al.}, who generalize to arbitrary causal structures in two steps. In both cases, they offer a solution for how to alter the description of these causal structures such that they fit in the framework of inflation and NPO.
We will adopt the first and alter the second method to adhere to our formalism.
One then has to show that these descriptions still obey all the results of the previous sections. 
Here we briefly outline why this is indeed the case and give the general transformation rules to map those causal structures to equivalent network scenarios.

\begin{figure}[tb!]
    \centering
    (a)
    \includegraphics[width=0.25\textwidth]{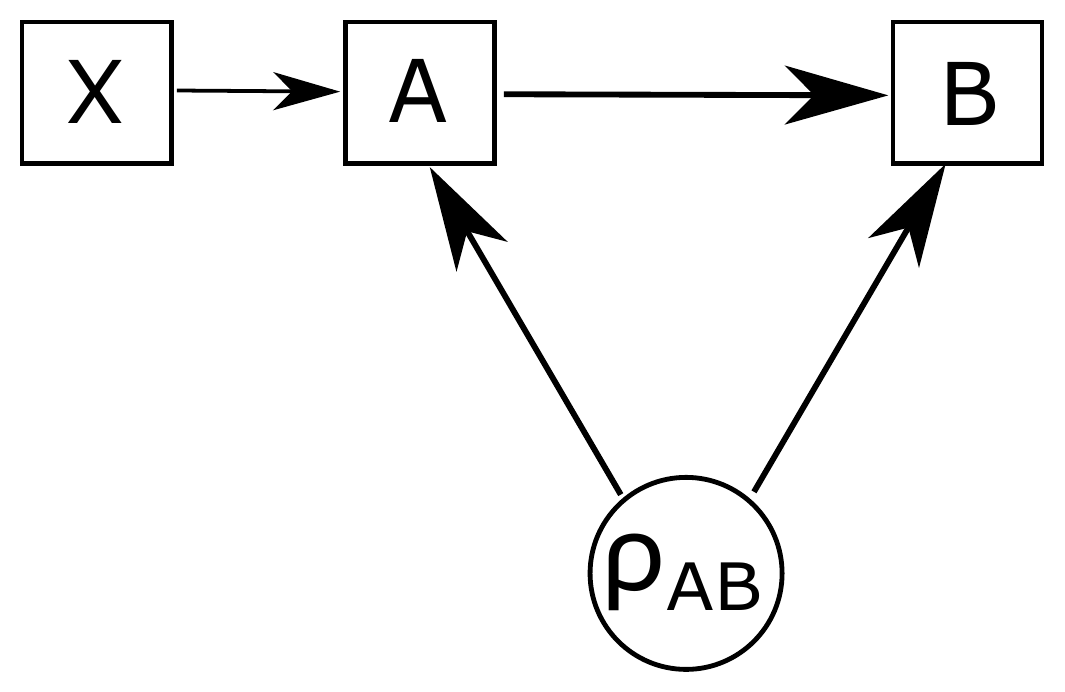} \hspace{2cm}
    (b)
    \includegraphics[width=0.25\textwidth]{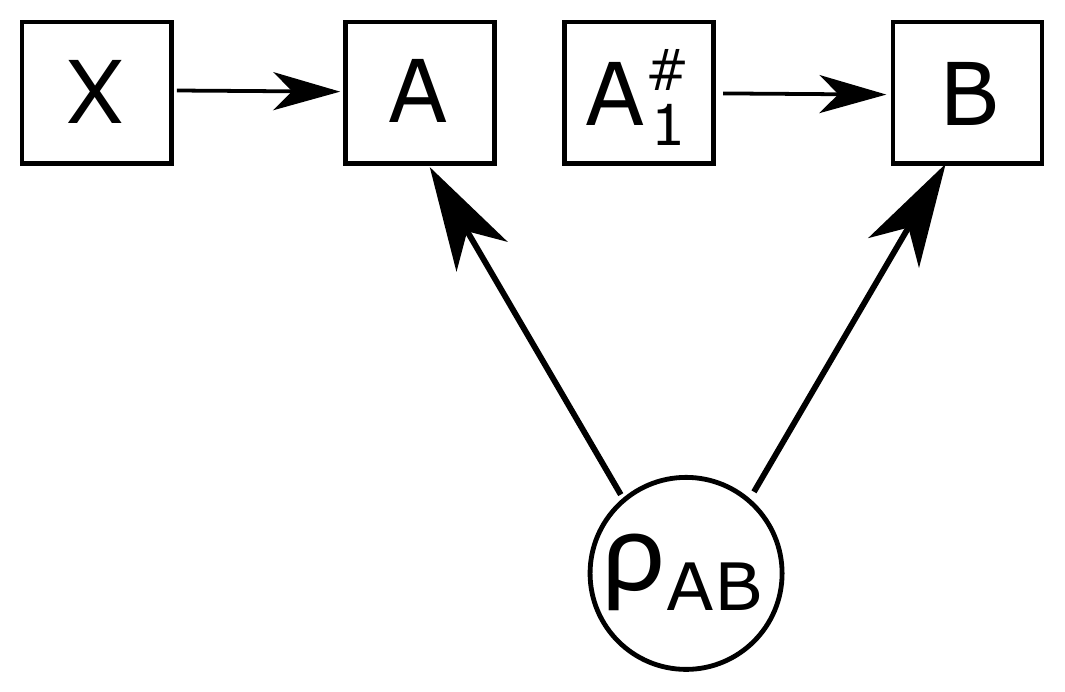}
    \caption{(a) The instrumental scenario as an example of a latent exogenous causal structure. The variable $A$ is both a parent and a child and thus the causal structure is not a network scenario. By splitting $A$ into $A$ and $A_1^\#$, as in (b), and post-selecting $A_1^\#$ on the outcome of $A$, the instrumental scenario can be modeled by a network scenario, which happens to be the Bell scenario. This process is an example of maximal interruption.}
    \label{fig:latent_exo}
\end{figure}

In the first generalization step, the causal structure is also allowed to contain observed nodes that have both one or more parents and one or more children. However, all unobserved nodes remain root nodes. Such causal structures are called \emph{latent exogenous}. 
An example is the instrumental scenario in Fig.~\ref{fig:latent_exo}. The probability distribution of the instrumental scenario is denoted by $P_{IS}(A,B,X)$. 
However, it is more common to express the statistics in terms of conditional probabilities. For a set of $N$ random variables $\{A_1, \ldots, A_N\}$ conditioned on $K$ independent variables $\{X_1, \ldots, X_K\}$, the statistics are fully captured by the combination of the conditional probabilities
\begin{align}
	P(A_1, \ldots, A_N \mid X_1 \ldots X_K) = \frac{P(A_1, \ldots, A_N, X_1 \ldots X_K)}{P( X_1 \ldots X_K)},
\end{align}
and the requirement that the setting-associated variables factorize:
\begin{align}
	P(X_1, \ldots, X_K) = P(X_1) \ldots P(X_K).
\end{align}
We refer to Ref.~\cite{richardson2014factorization} for a fuller discussion on why this is necessary and how this works in more general (classical) setups.
Hence, checking for compatibility will consist of two separate steps: First one has to confirm that variables that are being conditioned on form a product distribution; Secondly, one checks compatibility with the causal structure via the procedure outlined below.

In particular, for the instrumental scenario this simply reduces to checking compatibility of
\begin{align} \label{eq:conditional_prob}
    P_{IS}(a,b\mid x) = \frac{P_{IS}(a,b,x)}{P_{IS}(x)}. 
\end{align}
The classical random variable $A$ is both a child of $\rho_{AB}$ and $X$, as well as a parent of $B$. 
This problem is resolved by a process known as \emph{maximal interruption} \cite{gachechiladze2020quantifying, van2019quantum, agresti2019experimental,agresti2021experimental}: the variable $A$ is split into two random variables $A$ and $A_1^\#$, where $A$ is only a child and $A_1^\#$ is only a parent. The causal structure is then effectively mapped to the Bell scenario. By post-selecting $A_1^\#$ on the outcome $A=a$, i.e.\ by setting $P_{IS}(A=a, B=b\mid X=x) = P_{Bell}(A=a, B=b\mid X=x, A_1^\# = a)$, one can still obtain the allowed distributions of the original graph. 

\begin{figure}[tb!]
    \centering
    (a)
    \includegraphics[width=0.30\textwidth]{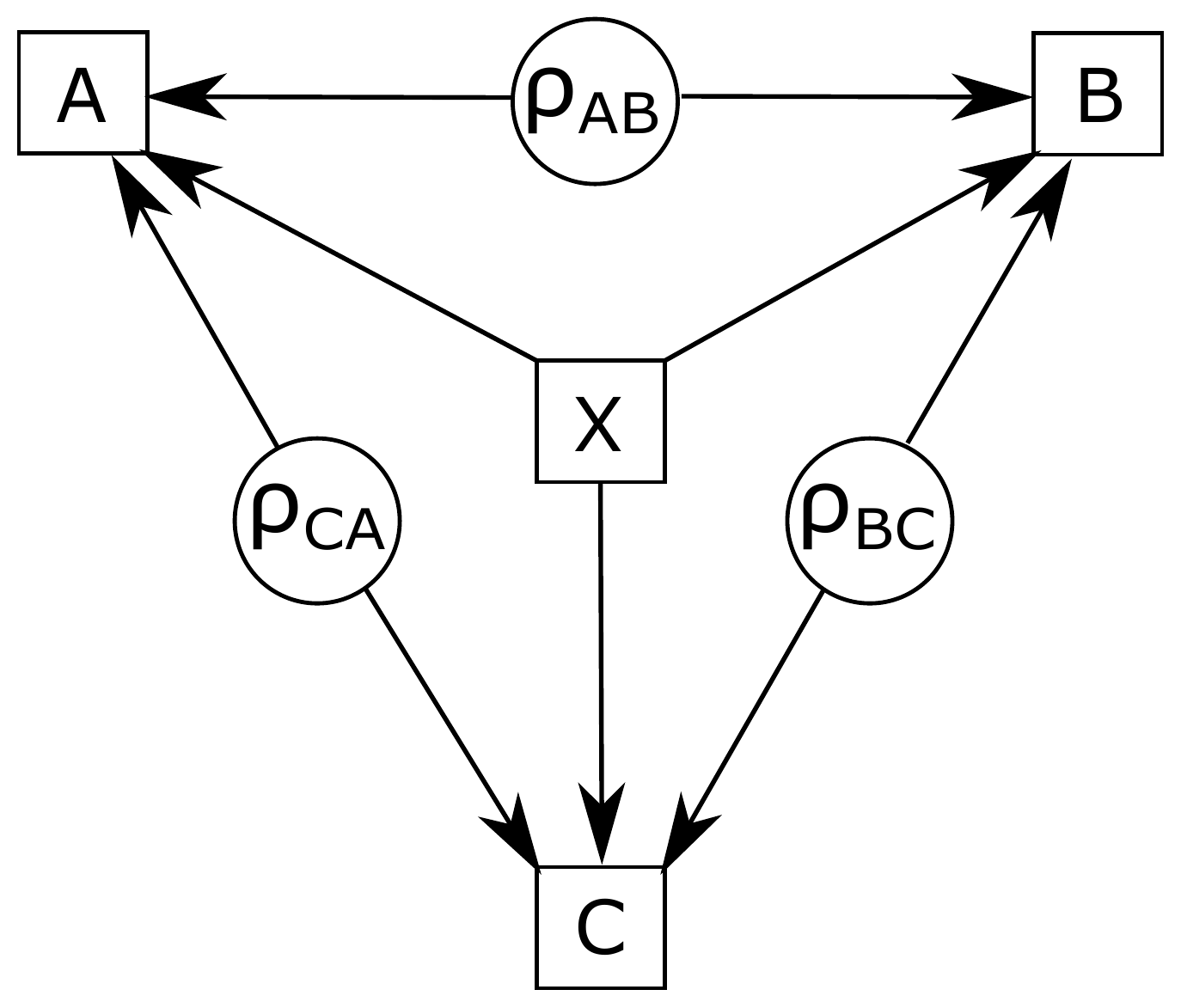} \hspace{2cm}
    (b)
    \includegraphics[width=0.30\textwidth]{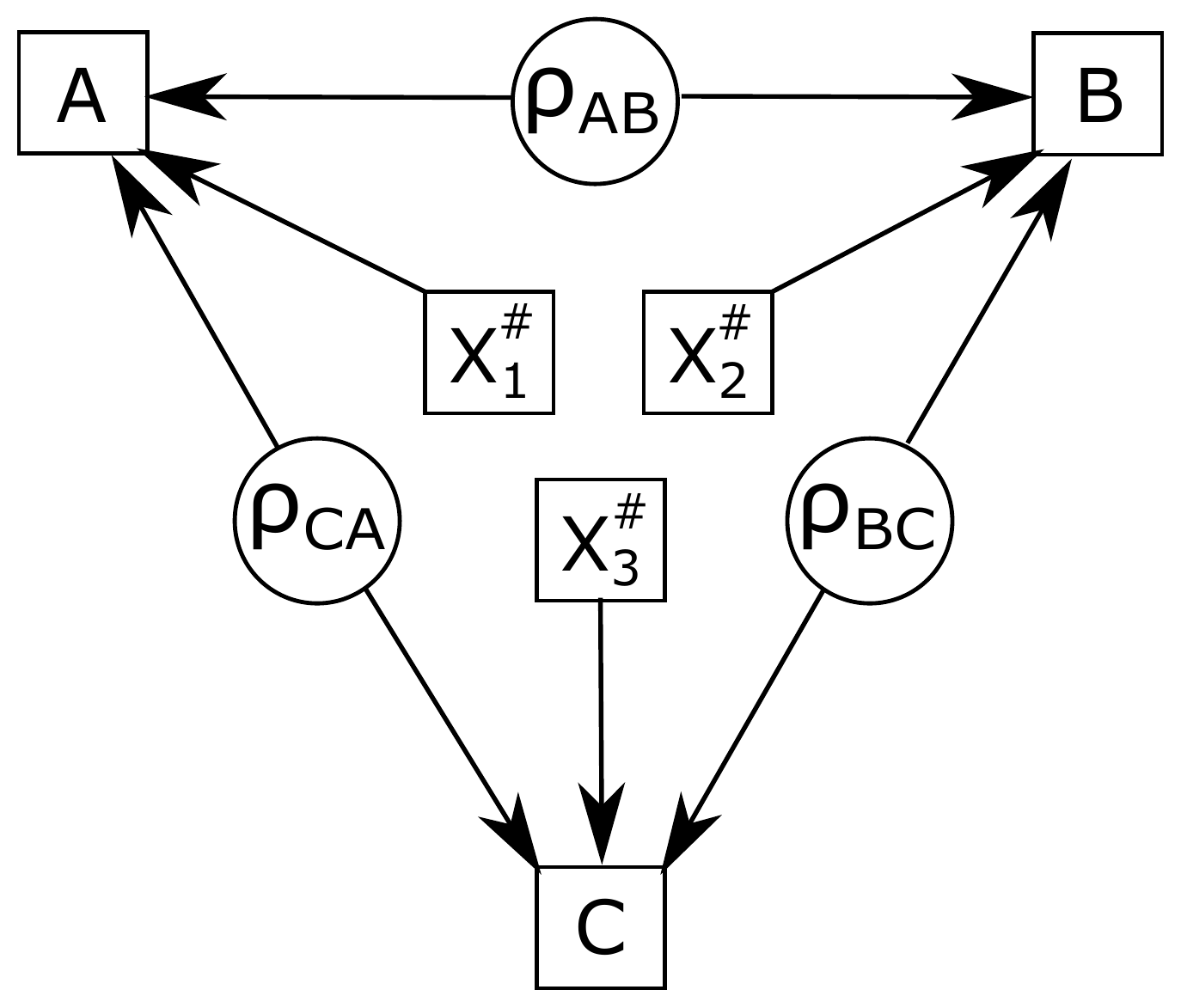}
    \caption{(a) The triangle scenario with a shared setting $X$. Via maximal interruption the observed variable $X$ is split into three independent and identically distributed random variables, producing the network scenario (b). The allowed probability distributions of the original causal structure can be obtained from the network scenario by post-selection.}
    \label{fig:triangle_shared}
\end{figure}

In the case that an observable node has multiple children, one applies a similar splitting and post-selection procedure. Consider for example the triangle scenario with a shared setting $X$ depicted in Fig~\ref{fig:triangle_shared}(a). 
Via maximal interruption the observed variable $X$ is split up into three independent and identically distributed random variables $X_1^{\#}$, $X_2^{\#}$ and $X_3^{\#}$ as depicted in Fig~\ref{fig:triangle_shared}(b). 
The resulting causal structure is a network scenario. By post-selecting on $X_1^{\#} = X_2^{\#} = X_3^{\#} = x$ one can link the (conditional) probability distribution of the network scenario to the (conditional) probability distribution $P(A,B,C\mid X=x)$ of the triangle with a shared setting and find the allowed distributions in the original causal structure.

In this way, one can map all latent exogenous causal structures to network scenarios.
The general rule is then as follows: 
Whenever an observed node is not a leaf node and is directly connected to multiple other nodes, split the node into as many copies as there are outgoing arrows. 
Remove every outgoing arrow from the original node and attach it to a copy. 
If there are no incoming arrows to the original node, remove it. 
Check whether the setting-associated variables factorize into a product distribution.
Finally, analyze this causal structure, which is now a network scenario, and apply post-selection on the copies.

\subsubsection{Non-exogenous causal structures}
The second step in the generalization is more involved. 
The causal structures are now also allowed to have latent (quantum) variables with parents.
In classical causal structures it is possible to transform these \emph{non-exogenous} causal structures into exogenous ones, for which it is then possible to apply maximal interruption as described above, if necessary \cite{evans2016graphs}. However, for quantum causal structures this is in general not possible. 

Ref.~\cite{wolfe2021quantum} gives a clear example (credited there to Stefano Pironio), which is depicted in Fig.~\ref{fig:latent_non_exo}. In this example, the quantum system $\rho_{BC}$ in Fig.~\ref{fig:latent_non_exo}(a) is non-exogenous. 
The structure in Fig.~\ref{fig:latent_non_exo}(b) has been exogenized as if it were a classical causal structure, by removing $\rho_{BC}$ and drawing arrows from the parents of $\rho_{BC}$ to its children. 
It can be seen that in the original causal structure it is possible to maximally violate a Bell inequality for either the systems $A$ and $B$ or $A$ and $C$, based on the setting determined by $S$. 
However, after the exogenization depicted on the right, the setting $S$ cannot determine anymore which pair maximally violates a Bell inequality. 
Since it is impossible that both $A$ and $B$, as well as $A$ and $C$ maximally violate a Bell inequality due to monogamy of entanglement \cite{koashi2004monogamy}, the causal structure on the right cannot produce the same statistics as the one on the left. 

We will split the treatment of non-exogenous causal structures in two parts: First, we will consider unobservable systems with one or more observed parents and at most one unobservable parent. 
Secondly, we will regard unobservable systems with multiple unobservable parents, but no observed parents. 
These two solutions can be combined to form a general set of rules for treating non-exogenous causal structures. 
It will turn out that, even though we cannot directly apply the exogenization procedure for classical causal structures, each observable leaf node will instead get an index associated to each of its parents in the classically exogenized causal structure. 
These indices then contain the information about commutation relations and independence constraints for the causal structure.

In Ref.~\cite{wolfe2021quantum} unobservable systems are eliminated by instead regarding such a system as a quantum channel applied to its unobservable parents and regarding any observed parents as a classical control for this quantum channel. If there is no observed parent to a non-exogenous system, the quantum channel that replaces it does not have such a control variable.

We opt for a slightly different treatment of non-exogenous systems with a similar interpretation. 
Instead of acting with a quantum channel on a state, we alter the POVM elements. Consider again the causal structure of Fig.~\ref{fig:latent_non_exo}(a) as an example. 
The intermediate state $\rho_{BC}$ has the interpretation of redistributing the $S$ subsystem of $\rho_{AS}$ among Bob and Charlie, based on the observed variable $S$. Hence, for different outcomes $s$ of $S$, Bob and Charlie will perform measurements over different parts of the $S$ subsystem of $\rho_{AS}$. 
However, for every specific outcome $s$, the measurements operators will be given by commuting POVMs for Bob and Charlie. We therefore define for every outcome of $S$ the commuting algebras $\Bcal^{(s)}$ and $\Ccal^{(s)}$ with elements $\{F_b^{(s)}\}_b$ and $\{G_c^{(s)}\}_c$ respectively.
Elements from these algebras obey the commutation relations
\begin{align}
    [F_b^{(s)},G_c^{(s)}] &= 0,
\end{align}
but whenever two operators have different indices $s$ and $s'$, they are no longer required to commute. The difference between this description and the measurement settings in network scenarios is that in network scenarios all POVM elements of Bob commute with all POVM elements of Charlie, while that is no longer the case here. We therefore propose a new graphical notation for this exogenization procedure, in which the nodes of Bob and Charlie are initially combined and only become commuting POVMs after the measurement setting $s$ has been processed (see Fig.~\ref{fig:latent_non_exo}(c)). We will call such nodes \emph{endogenous} nodes.

\begin{figure}[tb!] 
    \centering
         (a)
         \includegraphics[width=0.26\textwidth]{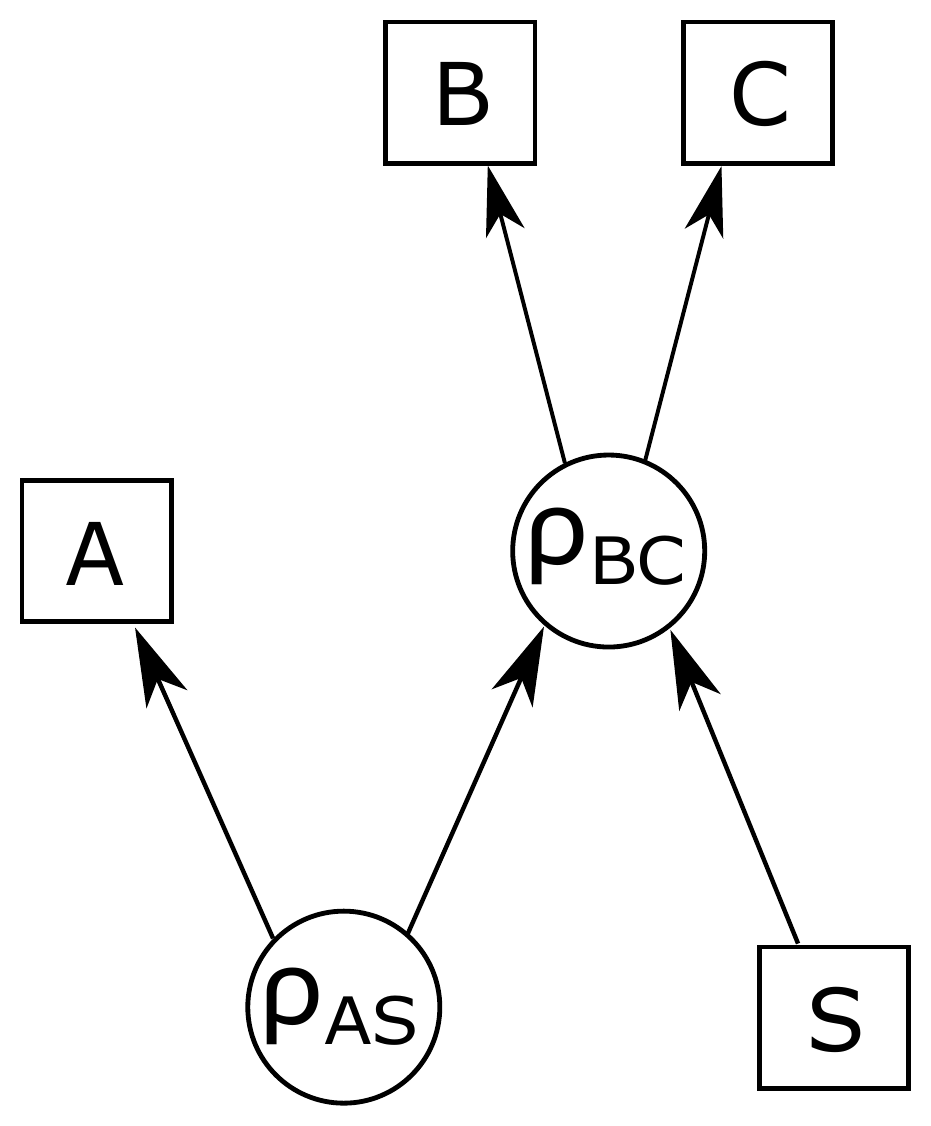} \hspace{1cm}
         (b)
         \includegraphics[width=0.24\textwidth]{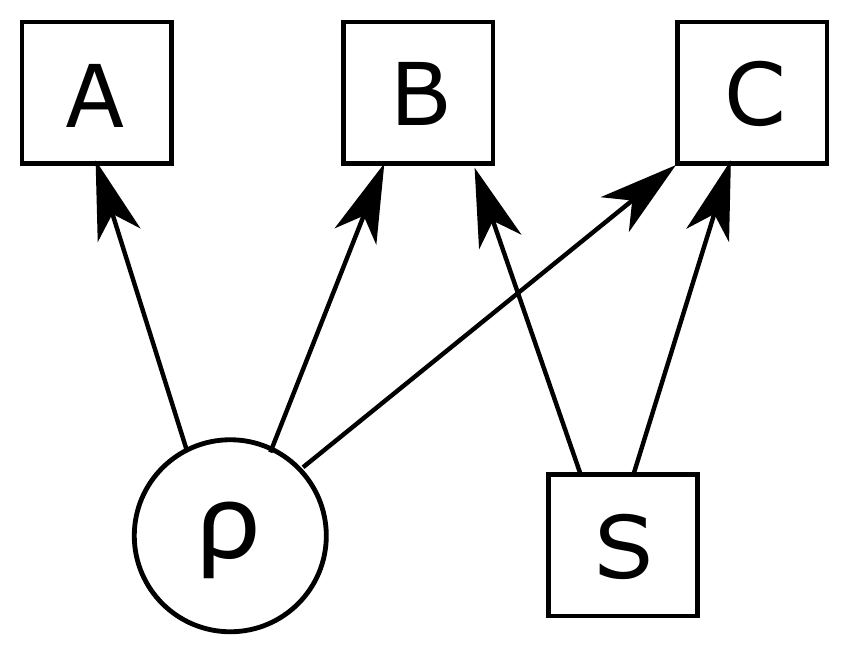}
         \hspace{1cm}
         (c)
         \includegraphics[width=0.21\textwidth]{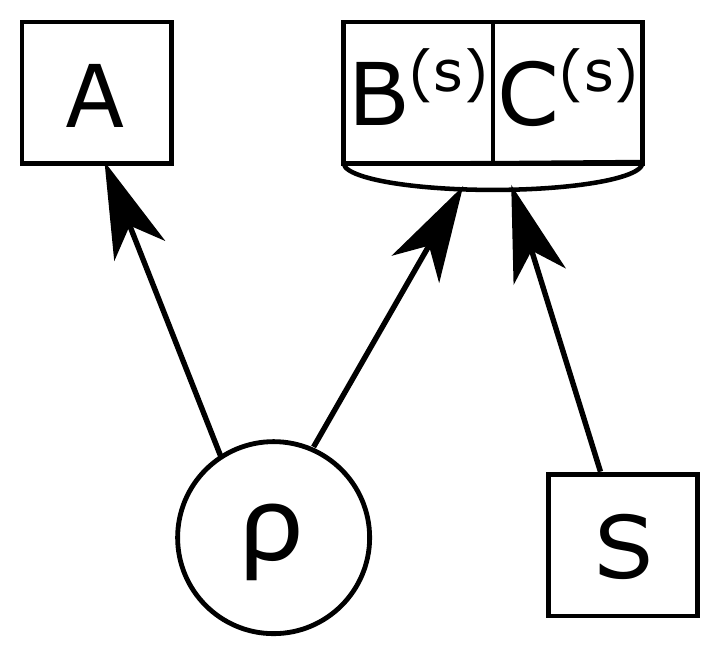}
    \caption{(a) An example of a non-exogenous causal structure. It is not possible to exogenize such a quantum causal structure as depicted in (b), like one would do for classical causal structures. This can be seen by noting that structure (a) can maximally violate a Bell inequality for $A$ and $B$ \emph{or} for $A$ and $C$ based on the measurement setting $S$. Due to monogamy of entanglement this is not possible for structure (b). In structure (c) this is solved by first regarding Bob and Charlie as one observer and only choosing the commuting algebras $\Bcal$ and $\Ccal$ after the setting $S$ has been received. This procedure is represented by a new type of node.}
\label{fig:latent_non_exo}
\end{figure}

One can still apply inflation to this causal structure as well. The algebras $\Acal$, $\Bcal^{(s)}$ and $\Ccal^{(s)}$ will be copied and get the index $i$ corresponding to the $i$'th copy of $\rho$. 
Different inflation levels will be modeled by commuting subalgebras with an exchange symmetry and one can show, using the de Finetti theorem, that in the limit a symmetric global state is separable across copies of $\rho$. The proofs of Theorem \ref{thm:hierarchy_convergence} and Corollary \ref{cor:causal_compatibility} then also follow. 

In general, one can apply the following rule to remove a non-exogenous system with observed parents and at most one latent parent, starting with non-exogenous systems that are closest to a leaf node: 
Split up each leaf node according to the structure of its local algebras, similar to the triangle scenario. Combine all leaf nodes that have a directed path from the non-exogenous system to that leaf node into one endogenous node. 
For every observed variable that is a parent of the non-exogenous node, introduce an index to the elements of the endogenous node. 
Elements of the algebra of the endogenous node commute if all such setting indices are the same and the elements originated from spatially separated systems (e.g. Bob and Charlie).

\hspace{0.2cm}

\noindent The final class of causal structures that has not been discussed yet, is the one where there are multiple unobserved parents to a latent variable. We will again first treat an example and then give the general rule.

Consider the causal structure in Fig.~\ref{fig:complicated_non-exogenous}(a). The intermediate node $\rho_{BC}$ is non-exogenous and has two latent parents. 
We start by splitting up the algebras $\Ccal$ and $\Dcal$ into their minus and plus sub-algebras, similar to the triangle scenario. Then we remove the non-exogenous node by taking those algebras together into an endogenous node that are its descendants, as was done in Fig.~\ref{fig:latent_non_exo}(c).
In this case, that will remove $\rho_{BC}$ and combine $\Bcal$ and $\Ccal_-$ into an endogenous node. 

When the causal structure is inflated, the root nodes are copied and given inflation indices $i,j,k$ respectively. 
The intermediate state $\rho_{BC}$ would then have gotten the two inflation indices $i,j$, which in turn are \emph{both} passed down to $\Bcal$ and $\Ccal_-$ (see Fig.~\ref{fig:complicated_non-exogenous}(b)). 
We thus have the following algebras after inflation: $\Acal^{i}, \Bcal^{ij}, \Ccal^{ij}_-, \Ccal^k_+, \Dcal^j_-$ and $\Dcal^k_+$. 
Operators from these algebras will be denoted in a similar way. 
Let $\Ecal^n$ be the algebra describing the level-$n$ inflated causal structure.
Though each of these algebras is a subalgebra of $\Ecal^n$, it is no longer true that $\Ecal^n$ is the tensor product of all these algebras, because some of them do not commute. 
This is due to the ``mixing'' of the root nodes $\rho_L^i$ and $\rho_M^j$ by the intermediate nodes $\rho^{ij}_{BC}$. 
In particular, for $B^{ij} \in \Bcal^{ij}$ and $C^{i'j'}_- \in \Ccal_-^{i'j'}$
\begin{align}
    [B^{ij}, C^{i',j'}_-] 
    \begin{cases}
    = 0 \qquad \text{if } i=i', j=j' \text{ or } i\neq i', j\neq j',\\
    \neq 0 \qquad \text{if } i=i', j\neq j' \text{ or } i\neq i', j=j'.
    \end{cases}
\end{align}
By requiring that the POVM elements of Bob and Charlie commute when they perform a measurement on the same state $\rho^{ij}_{BC}$, or on independent copies of the state, $\rho^{ij}_{BC}$ and $\rho^{i'j'}_{BC}$ with $i\neq i',\ j\neq j'$, we ensure that we still model spatially separated measurements in a physical scenario.

The question is now which independence relations hold and how to properly apply the quantum de Finetti theorem to show that these relations hold asymptotically in the inflation formalism.

To answer this question, note that the state $\rho_{BC}$ masks the independence of the parts of $\rho_L$ and $\rho_M$ that are sent to Bob and Charlie: The correlations in this causal structure could have also been produced by a four-partite state, of which the first and fourth subsystems are required to be independent, combined with the independent bi-partite state $\rho_R$.
The independence requirements that still have to hold after the quantum channel that produces $\rho_{BC}$ are thus
\begin{align}
    \label{eq:weird_independence_1}
    \rho(A^i B^{ij} C_-^{ij} C_+^k D_-^j D_+^k) &= \rho(A^i B^{ij} C_-^{ij} D_-^j) \rho(C_+^k D_+^k) \qquad\ \ \forall i,j,k,\\
    \label{eq:weird_independence_2}
    \rho(A^i D_-^j) &= \rho(A^i) \rho(D_-^j), \qquad \qquad \qquad \qquad \forall i,j\\
    \label{eq:weird_independence_3}
    \rho(\prod_{i=1}^n A^i B^{ii} C_-^{ii} C_+^i D_-^i D_+^i) &= \rho(\prod_{i=1}^n A^{\pi(i)} B^{\pi(i)\pi'(i)} C_-^{\pi(i)\pi'(i)} C_+^{\pi''(i)} D_-^{\pi'(i)} D_+^{\pi''(i)}) = \prod_{i=1}^n \rho(A^i B^{ii} C_-^{ii} C_+^i D_-^i D_+^i),
\end{align}
where the first two equalities signify independence within layers of inflation (namely the independence of $\rho_R$ with respect to $\rho_L$ and $\rho_M$, and the independence of $A$ and $D$), while the last equality corresponds to independence between layers of inflation.
Note, however, that the algebras $\Bcal^{ij}$ and $\Ccal^{ij}_-$ do not have this independence between inflation layers if the state is evaluated over products of operators of either algebra for which only one of the two indices coincides, e.g.\ $\rho(B^{ij} C_-^{i'j}) \neq \rho(B^{ij})\rho(C_-^{i'j})$ and $\rho(B^{ij} B^{i'j}) \neq \rho(B^{ij})\rho(B^{i'j})$.

Relaxing the independence conditions of Eqs.~\eqref{eq:weird_independence_1}-\eqref{eq:weird_independence_3} to their corresponding symmetry constraints and combining this with the commutation relations and the de Finetti Theorem, we can still asymptotically solve the (rank-constrained) causal compatibility problem even for this causal structure.

The general rule to map each inflation level of a causal structure with non-exogenous systems that have multiple latent parents, but no observed parents to an equivalent latent exogenous causal structure is then as follows: 
Start with the non-exogenous system that is closest to a leaf node. 
Split up each leaf node according to the structure of its local algebras. 
Combine all leaf nodes that have a directed path from the non-exogenous system to that leaf node into one endogenous node. 
For every root node that is an ancestor of the non-exogenous node, attach the inflation index of that root node to the elements of the endogenous node.
Elements of the algebra of the endogenous node commute if either (1) all inflation indices are pair-wise the same, i.e.\ $i=i', j=j',\ldots$, and the elements originated from spatially separated systems, e.g.\ Bob and Charlie, or if (2) all inflation indices are pair-wise different, i.e.\ $i\neq i', j\neq j',\ldots$.

\hspace{0.2cm}

\noindent
More succinctly: to reduce a non-exogenous causal structure to an exogenous one, we attach to every child of a non-exogenous node an index for all of the root nodes of the non-exogenous system, and apply the appropriate commutation relations.

Our approach is thus applicable to all relevant quantum causal structures, by sequentially applying the techniques outlined above, i.e.~by first turning a non-exogenous causal structure into an exogenous one with a new type of node that can also be treated in our model, and then applying maximal interruption to turn it into a network scenario. Checking the feasibility of the SDP hierarchy of this network scenario, as well as the required factorization of the setting-associated variables is then a necessary and sufficient procedure for checking causal compatibility.

\begin{figure}
    \centering
    (a)
    \includegraphics[width=0.3\linewidth]{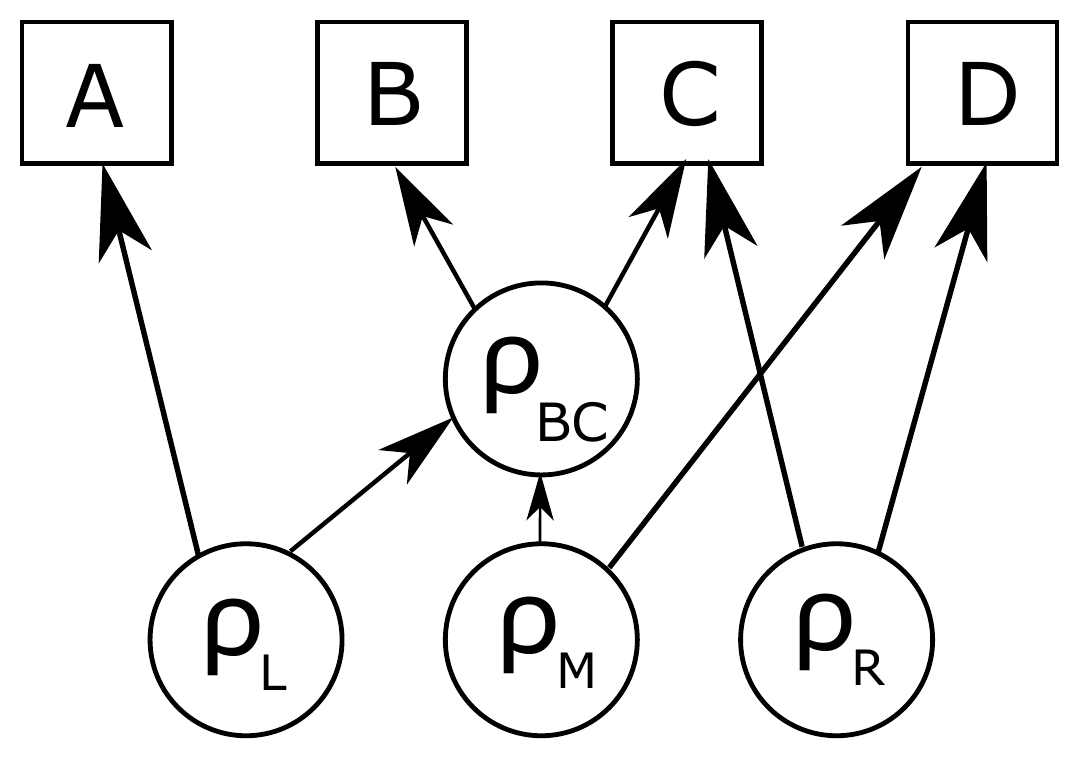}
    \hspace{2cm}
    (b)
    \includegraphics[width=0.32\linewidth]{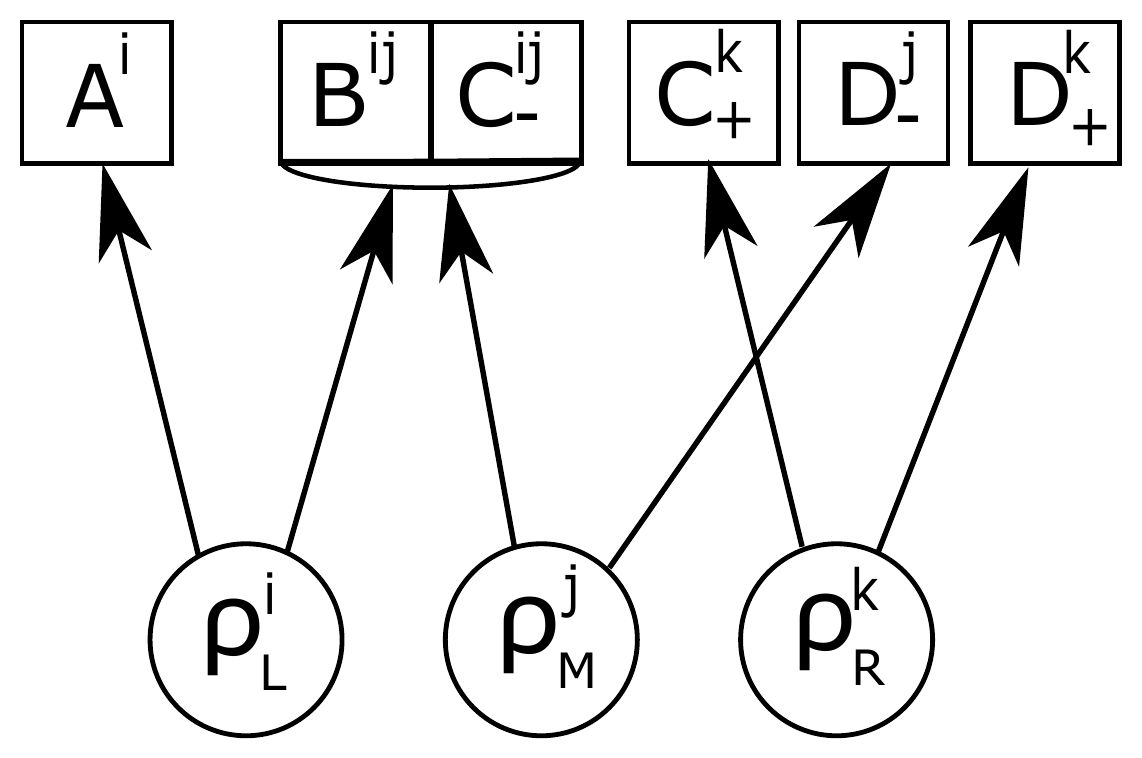}
    \caption{(a) A more complicated causal structure, where $\rho_{BC}$ has two latent parents and two observable children. The inflated version of this structure can alternatively be depicted as in Fig. (b). The algebras $\{\Bcal^{ij}\}_{ij}$ and $\{\Ccal^{ij}_-\}_{ij}$ are taken together and do not commute if there is only one overlapping index.}
    \label{fig:complicated_non-exogenous}
\end{figure}

\section{Conclusions and Outlook} \label{sec:Conclusions}

Building on the quantum inflation hierarchy of Ref.~\cite{wolfe2021quantum}, we have constructed a provably complete semidefinite programming hierarchy for the quantum causal polynomial optimization problem. 
Along the way, we have generalized the Quantum de Finetti Theorem for infinite systems \cite{raggio1989quantum} to arbitrary $C^*$-tensor products, 
and given a description of the NPO hierarchy \cite{pironio2010convergent} as an optimization procedure over states of a universal $C^*$-algebra.

A number of follow-up questions suggest themselves.

We could not prove completeness of the original quantum inflation hierarchy, due to the difficulty of constructing the local observable algebras ($\Acal_-, \Acal_+$, etc.).
To deal with this problem, we had to manually add generators for the algebra to the NPO program, and then manually impose norm constraints.
While we have argued that any constructive completeness proof will have to add elements to the algebra that is extracted from the output of the SDP hierarchy, it is not obvious that we have found the most economical way of handling the issue.
We also do not know whether there are a priori finite bounds on the norm of the local operators that combine to give the POVM elements.
Both questions merit further researcher.

We have focused mostly on the quantum causal compatibility problem, and have not investigated objective functions beyond the 2-norm distance to measured data.
While the examples given in \cite[Section~VII]{wolfe2021quantum} carry over to our formulation, it would be interesting to look into further applications.

It would be worth investigating what further properties of the universal observable algebra can be enforced with suitable constraints.
For example, one could impose that some of the subalgebras are Abelian in order to model partly classical behavior.

Numerical results for the SDP hierarchy described in this paper will be presented elsewhere.

\section{Acknowledgments}

We thank 
Johan \AA berg,
Nikolai Miklin,
Marc-Olivier Renou,
Reinhard Werner, and
Elie Wolfe
for insightful discussions. 
This work has been supported by Germany's Excellence Strategy -- Cluster of Excellence Matter and Light for Quantum Computing (ML4Q) EXC 2004/1 -- 390534769.

{\bf Data sharing.} Data sharing is not applicable to this article as no data sets were generated or analysed during the current study.

\bibliographystyle{sn-mathphys}
\bibliography{references}

\end{document}